\documentclass[11pt]{amsart}
\usepackage{amsmath,amsxtra,amssymb,bm,color,latexsym,epsfig,amscd,amsthm,subfigure,fancybox,epsfig}
\usepackage[mathscr]{eucal}
\usepackage{bbm}
\usepackage{float}
\usepackage{graphicx}
\usepackage{enumerate}
\usepackage{enumitem}
\usepackage{epsfig}
\usepackage{epstopdf}
\usepackage{cases}
\usepackage{hyperref}
\usepackage{tikz}
\usetikzlibrary{intersections}
\usetikzlibrary{positioning}
\usetikzlibrary{graphs}

\hypersetup{
	colorlinks=true,
	linkcolor=blue,
	filecolor=magenta,
	urlcolor=cyan,
	citecolor=blue,
}

\newtheorem{thm}{Theorem}[section]
\newtheorem {asp}{Assumption}[section]
\newtheorem{lm}{Lemma}[section]

\theoremstyle{definition}

\theoremstyle{remark}
\newtheorem{rem}{Remark}[section]

\numberwithin{equation}{section}

\newcommand{\eps}{\varepsilon}

\newcommand{\M}{\mathcal{M}}

\newcommand{\E}{\mathbb{E}}

\newcommand{\Lom}{\mathcal{L}}

\newcommand{\Z}{\mathbb{Z}}

\newcommand{\PP}{\mathbb{P}}

\newcommand{\R}{\mathbb{R}}

\numberwithin{equation}{section}
\newcommand{\1}{\boldsymbol{1}}

\newcommand{\wdt}{\widetilde}

\newcommand{\bed}{\begin{equation}}
	\newcommand{\eed}{\end{equation}}
\newcommand{\bea}{\bed\begin{array}{rl}}
	\newcommand{\eea}{\end{array}\eed}

\newcommand{\barray}{\begin{array}{ll}}
	\newcommand{\earray}{\end{array}}
\newcommand{\diag}{{\rm diag}}

\def\disp{\displaystyle}

\def\bar{\overline}
\def\hat{\widehat}
\def\a.s{\text{\;a.s.\;}}

\def\bnu{\boldsymbol{\nu}}
\def\bmu{\boldsymbol{\nu}}

\newcommand{\Se}{\overline{\mathcal{S}}}

\newcommand{\wtau}{\widetilde\tau}

\begin{document}
	\title{Long-term behavior of stochastic SIQRS epidemic models}
	\author[A. Hening]{Alexandru Hening }
	\address{Department of Mathematics\\
		Texas A\&M University\\
		Mailstop 3368\\
		College Station, TX 77843-3368\\
		United States
                }
	\email{ahening@tamu.edu}

	\author[D.H. Nguyen]{Dang H. Nguyen }
	\address{Department of Mathematics \\
		University of Alabama\\
		345 Gordon Palmer Hall\\
		Box 870350 \\
		Tuscaloosa, AL 35487-0350 \\
		United States
                }
	\email{dangnh.maths@gmail.com}
	
        \author[Trang Ta]{Trang Ta}
	\address{Department of Mathematics \\
                 University of Alabama\\
                 345 Gordon Palmer Hall\\
                 Box 870350 \\
                 Tuscaloosa, AL 35487-0350 \\
                 United States
                 }
        \email{ttta@crimson.ua.edu}
	
	\author[S. Ungureanu]{Sergiu C. Ungureanu}
        \address{Department of Economics\\
                 City, University of London\\
                 Northampton Square\\
                 London EC1V 0HB\\
                 United Kingdom
                 }
        \email{Sergiu.Ungureanu.1@city.ac.uk}

\date{}
	
	
\begin{abstract}
In this paper we analyze and classify the dynamics of SIQRS epidemiological models with susceptible, infected, quarantined, and recovered classes, where the recovered individuals can become reinfected.
We are able to treat general incidence functional responses.
Our models are more realistic than what has been studied in the literature since they include two important types of random fluctuations.
The first type is due to small fluctuations of the various model parameters and leads to white noise terms.
The second type of noise is due to significant environment regime shifts in the that can happen at random.
The environment switches randomly between a finite number of environmental states, each with a possibly different disease dynamic.
We prove that the long-term fate of the disease is fully determined by a real-valued threshold $\lambda$.
When $\lambda < 0$ the disease goes extinct asymptotically at an exponential rate.
On the other hand, if $\lambda > 0$ the disease will persist indefinitely.
We end our analysis by looking at some important examples where $\lambda$ can be computed explicitly, and by showcasing some simulation results that shed light on real-world situations.
\bigskip

\end{abstract}
	
	\maketitle
	
	\section{Introduction}\label{sec:int}
	Severe pandemics of infectious diseases have historically been, and will likely continue to be, a serious worldwide systemic risk.
	This risk needs to be understood and controlled. The policy debates around the most recent Covid-19 outbreak are a reminder that we need to have better and more realistic models.
	In this paper we analyze realistic stochastic epidemic models that can help to better understand and predict the evolution of outbreaks.
	The main goal is to be able to offer a straightforward and easy way to make policy recommendations. 
	
	The history of the mathematical modelling of epidemic models is long.
	Kermack and McKendrick, considered the fathers of mathematical epidemiology, were the first to analyze rigorously deterministic epidemic models - this happened in series of papers called ``Contributions to the Mathematical theory of epidemics" \cite{kermack1927contribution, kermack1932contributions}. In the first paper written in 1927, they assumed that complete immunity is conferred by a single attack, and that an individual is not infective before the moment he gets infected.
	They supposed that the population is subdivided into three distinct classes: the susceptible class ($S$), infective class ($I$) and the recovered, or removed, class ($R$).
	Individuals will be transferred consecutively between the classes, $S\to I\to R$. This kind of model is known in the literature as an SIR model.\footnote{A person who recovered will have permanent immunity in this type of model.}
	The dynamics can be modeled by the system of ordinary differential equations
	\begin{equation}
		\begin{cases}
			\begin{aligned}
				dS(t)&=\big(A-I(t)F\big(S(t), I(t)\big)-\mu S(t)+\gamma I(t)\big)dt,\\
				dI(t)&=I(t)\big(F\big(S(t), I(t)\big)-(\mu+\gamma_1+\gamma_2)\big)dt,\\
				dR(t)&=\big(\gamma_2 I(t)-(\mu+\gamma)R(t)\big)dt,
			\end{aligned}
		\end{cases}
	\end{equation}
	where $A>0$ is the birth rate of the population, $\mu$ is the independent death rate of all the groups, $\gamma_1>0$ is the death rate of infected individuals, $\gamma_2>0$ is the recovery rate of infected individuals, $\gamma>0$ is the rate of loss of immunity in the recovered group, and $I(t)F\big(S(t), I(t)\big)$ indicates the incidence rate, which describes the number of new cases per unit of time.
	There are many different types of incidence rates in the literature -- see \cite{nguyen2020long, phu2020longtime, huang2009global,dieu2016classification,du2017permanence,du2019conditions,jiang2016dynamical, liu2018analysis, nguyen2020analysis, nguyen2019stochastic, nie2018dynamic, ruan2003dynamical, zhang2013sirs, zhang2015threshold, tang2008coexistence} and the references therein. 
	
	It is not always the case that SIR models work well, as some of their assumptions could get violated. For some diseases, like syphilis or Covid-19 \cite{M21}, recovered patients can become susceptible again.
	In order to have more realistic models which can capture such diseases people have started looking at SIRS epidemic models where one can have dynamics of the form $S\to I\to R\to S$.
	During the Covid pandemic governments made use of quarantines, where infected individuals were separated in order to lessen the spread of the pandemic to the susceptible population.
	Quarantines have been widely used throughout the ages so it is natural to consider models where there is a quarantined class $Q$ of individuals.
	Quarantine leads to the isolation of infected individuals in order to try and contain the spread of the disease.
	We will assume that isolation makes it impossible for the quarantined to get in contact with susceptible individuals.
	As a result the transmissions of the infection to the susceptible class is reduced.
	This method has been used for thousands of years to reduce the transmission of diseases among human beings, from leprosy to the plague, TBC, Ebola, smallpox, etc.
	Models that include quarantine were thought to be appropriate for childhood diseases, where quarantine seems to destabilize the epidemic. However, this is sometimes detrimental as the destabilization can lead to oscillations and recurrent outbreaks of childhood diseases \cite{FT95}. 
	There are various ways to model the individuals who recover from a disease.
	In certain diseases, the recovered individuals are still susceptible to infection (SIQS models), while in others, recovered individuals are better modeled as permanently immune to the infection (SIQR models).
	SIQRS models combine features of these models for more generality, letting parametrisation decide the correct evolution of the susceptible and recovered categories.
	
	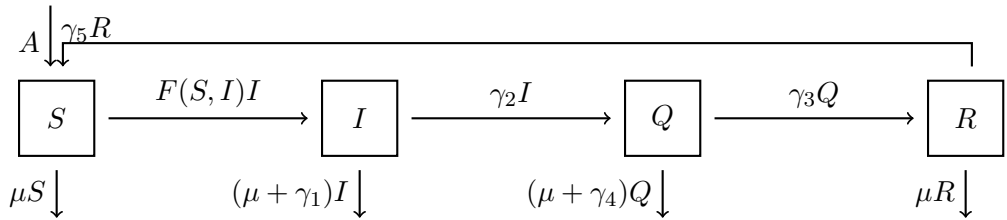
\begin{figure}
		\begin{center}
			\begin{tikzpicture}[
				skip loop/.style={to path={-- ++(0,#1) -| (\tikztotarget)}},
				thick,
				nonterminal/.style = {
					rectangle,			
					minimum size = 10mm,			
					thick,
					draw = black,
				}
				]
				\matrix [row sep=10mm, column sep=15mm] {
					\node (pa){}; &&&&&&\\
					\node (s) [nonterminal] {$S$}; &&
					\node (i) [nonterminal] {$I$}; && 
					\node (q) [nonterminal] {$Q$}; && 
					\node (r) [nonterminal] {$R$}; \\
					\node (pm1){}; && 
					\node (pmg1){}; && 
					\node (pmg4){}; && 
					\node (pm2){}; \\ 
				};
				\path (pa) edge[->, shorten >=5pt, transform canvas={xshift=-0.5ex}] node[left]{$A$} node[right, transform canvas={yshift=1ex}]{$\gamma_5 R$} (s);
				\path (s) edge[->, shorten <=5pt, shorten >=5pt] node[above]{$F(S,I) I$} (i);
				\path (i) edge[->, shorten <=5pt, shorten >=5pt] node[above]{$\gamma_2 I$} (q);
				\path (q) edge[->, shorten <=5pt, shorten >=5pt] node[above]{$\gamma_3 Q$} (r);
				\path (r) edge[->, shorten <=5pt, shorten >=5pt, transform canvas={xshift=0.5ex}, skip loop=10mm ] node[right]{} (s);
				\path (s) edge[->, shorten <=5pt, shorten >=5pt] node[left]{$\mu S$}(pm1);
				\path (i) edge[->, shorten <=5pt, shorten >=5pt] node[left]{$(\mu + \gamma_1) I$}(pmg1);
				\path (q) edge[->, shorten <=5pt, shorten >=5pt] node[left]{$(\mu + \gamma_4) Q$}(pmg4);
				\path (r) edge[->, shorten <=5pt, shorten >=5pt] node[left]{$\mu R$}(pm2);
			\end{tikzpicture}
		\end{center}
		\caption{Schematic representation of the system \eqref{ww0} and of the nonstochastic effects on the population groupings from \eqref{ww1}.}
		\label{fig:syst}
	\end{figure}
	
	The SIQRS dynamics can be modelled, in the absence of noise, by the system of ordinary differential equations
	\begin{equation}\label{ww0}
		\begin{cases}
			\begin{aligned}
				dS(t)&=\Big[A-F\big(S(t),I(t)\big)I(t)-\mu S(t)+\gamma_5 R(t)\Big]dt,\\
				dI(t)&=\Big[F\big(S(t),I(t)\big)-\mu-\gamma_1-\gamma_2\Big]I(t)dt,\\
				dQ(t)&=\Big[\gamma_2 I(t)-\left(\mu+\gamma_3+\gamma_4\right)Q(t)\Big]dt,\\
				dR(t)&=\Big[\gamma_3 Q(t)-\left(\mu+\gamma_5\right)R(t)\Big]dt.
			\end{aligned}
		\end{cases}
	\end{equation}
	
	Here, $A$ is again the population growth rate and $\mu$ is the independent death rate, common to all groups.
	$\gamma_1$ is the death rate specific to the infected category, attributed to the disease, $\gamma_2$ is the rate of transfer of the infected into the quarantined category, $\gamma_3$ is the rate of recovery of the quarantined individuals, $\gamma_4$ is the rate of death in the quarantined group, also attributed to the disease, and finally $\gamma_5$ is the rate of transfer of the recovered back to the susceptible group.
	If this last parameter is 0, then the recovered stay immune.
	The term $F\big(S,I\big)I(t)$ is the incidence rate and it describes the number of new infections per unit of time.
	$F\big(S,I\big)$ is the incidence rate per infective individual.
	In the most common model $F\big(S,I\big)I =\beta I S$, but it turns out that it is very important to consider more general nonlinear functions $F$. 
	
	There are many plausible mechanisms that can lead to nonlinear incidence rates.
	Many common pathogens cannot live separated from the host for long periods of time.
	In addition, susceptible individuals can usually be infected only once the concentration of the pathogen reaches a certain level.
	At low densities of infected individuals, low $I$, the threshold for infectivity may not be reached often.
	If $I$ increases, the threshold will be reached more often and infections will rise faster.
	If the disease is vectored and the vector must attack on average $q$ infective individuals in order to make its next attack infective, one can show that the incidence rate will be proportional to $I^q$ when one makes the natural assumption that the attack times in a given time period follow a Poisson distribution.
	Nonlinear incidence rates also arise naturally in models of soil-transmitted parasitic worm infections -- helminth infections -- when the probability of pairing is taken into account \cite{LLI86}.
	As showcased by \cite{LLI86}, SIR and SIRS models exhibit very different behavior if the incidence rate is not $F\big(S,I\big)I =\beta I S$, i.e., it is nonlinear.
	The dynamical behavior can include Hopf, saddle-node and homoclinic bifurcations. 
	
	Populations are influenced by random environmental fluctuations, and these fluctuations have a significant impact on growth dynamics.
	In order to have robust dynamic models, it is important to include these fluctuations in the model.
	Furthermore, environmental fluctuations can have a significant impact on whether a species survives in the long-term.
	For example, in certain settings, persistence can be reversed into extinction by environmental fluctuations, while in others, extinction can be transformed into persistence \cite{BL16, HN20, HNS21}.
	Common ways to model environmental fluctuations in population dynamics are systems of stochastic differential equations and Markov processes \cite{C82, ERSS13, EHS15,  LES03, SLS09, SBA11, BS09, BHS08, B18, HNC20}.
	A typical approach is to transform ordinary differential equations (ODE) models into stochastic differential equations (SDE) models.
	This boils down to saying that the various rates in the disease ecosystem are not constant, but fluctuate around their average values according to white noise.
	There is a well established general theory of coexistence and extinction for systems that come from ecology, when they are in Kolmogorov form \cite{SBA11,hening2018coexistence,HNC20}.
	
	In \cite{H02} the authors study the role of quarantine in a series of SIQR models. They show that `After examining many parameter sets and the corresponding periodic solutions, we have found that as one aspect becomes more realistic, another becomes less realistic. Thus we have not been able to find a parameter set that matches every feature of the observed data.' Nevertheless, they follow \cite{FT00} and say that stochastic models might be significantly more realistic.  
	
	It is very important to use stochastic models in epidemiological modelling, when possible. From the start it is natural to define the probability of disease transmission between two individuals instead of stating certainly that transmission will or will not happen. Furthermore, as the authors of \cite{AB12} explain, when one considers the extinction of endemic diseases, this phenomenon can only be analyzed using stochastic models since `extinction occurs when the epidemic process deviates from the expected level'. 
	
	In \cite{nguyen2020long} the authors consider a model which is perturbed by white noise and has the general incidence rate $S(t)I(t)/F\big(S(t), I(t)\big)$, where $F\big(S(t), I(t)\big)$ is a locally Lipschitz continuous function in both variables.
	They show that if the real-valued threshold $\lambda<0$, then the disease will eventually disappear, and if $\lambda>0$ the epidemic will become permanent.

	In certain systems, it makes more sense to assume that when the environment changes, the dynamics also changes significantly.
	In a deterministic setting this can be modelled by periodic vector fields which can be interpreted to mimic seasonal fluctuations.
	In the random setting, these types of fluctuations are captured by piecewise deterministic Markov processes (PDMP) -- see \cite{D84} for an introduction to PDMP.
	In a PDMP, the environment switches between a fixed finite number of states to each of which we associate an ODE.
	In each state the dynamics is given by the flow of its associated ODE.
	After a random time, the environment switches to a different state, and the dynamics is governed by the ODE from that state.
	
	Our goal is to capture both white noise and discrete types of fluctuations.
	As such we will look at SSDE (stochastic differential equations with switching).
	These processes involve a discrete component that keeps track of the environment and which changes at random times.
	In a fixed environmental state the system is modelled by stochastic differential equations.
	This way we can capture the more realistic behaviour of two types of environmental fluctuations: (1) large changes of the environment (daily or seasonal changes) and (2) small fluctuations within each environment.
	Random switching has been used in epidemiological settings \cite{lan2019stochastic, lin2014threshold, WL23}.
	In \cite{phu2020longtime} the authors considered the system
	\begin{equation}\label{ww3}
		\begin{cases}
			\begin{aligned}
				dS(t) &= \big[\Lambda\big(r(t)\big)-F\big(S(t), I(t),r(t)\big)I(t)-\mu\big(r(t)\big)S(t)+\gamma\big(r(t)\big)I(t)\big]dt\\
				&+\sigma_1\big(r(t)\big)S(t)dW_1(t),\\
				dI(t) &= I(t)\big[F\big(S(t), I(t),r(t)\big)-\wdt\gamma\big(r(t)\big)\big]dt\\
				&+\sigma_2\big(r(t)\big)I(t)dW_2(t),
			\end{aligned}
		\end{cases}
	\end{equation}
	
	\noindent
	where $\sigma_1, \sigma_2\neq 0$ parametrise the variance of the white noise, $\wdt\gamma = \mu_1 + \gamma$, where $\mu_1$ is a separate death rate for the infected, and $r(t)$ is a Markov chain taking values in a finite space.
	
	Our goal is to study \eqref{ww0} when there are both types of environmental fluctuations present.
	In its stochastic form, system \eqref{ww0} becomes
	\begin{equation}\label{ww1}
		\begin{cases}
			\begin{aligned}
				dS(t)&=\Big[A\big(r(t)\big)-F\big(S(t),I(t),r(t)\big)I(t)-\mu\big(r(t)\big)S(t)+\gamma_5\big(r(t)\big)R(t)\Big]dt\\
				&+\sigma_1\big(r(t)\big)S(t)dW_1(t),\\
				dI(t)&=\Big[F\big(S(t),I(t),r(t)\big)-\mu\big(r(t)\big)-\gamma_1\big(r(t)\big)-\gamma_2\big(r(t)\big)\Big]I(t)dt\\
				&+\sigma_2\big(r(t)\big)I(t)dW_2(t),\\
				dQ(t)&=\Big[\gamma_2\big(r(t)\big)I(t)-\big(\mu(r(t))+\gamma_3(r(t))+\gamma_4(r(t))\big)Q(t)\Big]dt\\
				&+\sigma_3\big(r(t)\big)Q(t)dW_3(t),\\
				dR(t)&=\Big[\gamma_3(r(t))Q(t)-\big(\mu\left(r(t)\right)+\gamma_5\left(r(t)\right)\big)R(t)\Big]dt\\
				&+\sigma_4\big(r(t)\big)R(t)dW_4(t),\\
			\end{aligned}
		\end{cases}
	\end{equation}
	
	\noindent
	where $\sigma_i\neq 0$ for $i=1,\dots,4$, and the other parameters are interpreted as for \eqref{ww0}.
	
	Here $\{r(t), t\geq 0\}$ is a right-continuous Markov chain on $\big(\Omega, \mathcal F, \{\mathcal F_t\}_{t\geq 0}, \mathbb P\big)$, taking values in the finite state space $\mathcal M:=\{1,\dots, m_0\}$.
	The Markov chain $\{r(t), t\geq 0\}$ is supposed to have an irreducible generator $\mathbf{Q}=(q_{kl})_{m_0\times m_0}$.
	This implies in particular that $\{r(t), t\geq 0\}$ has a unique invariant probability measure $\pi=(\pi_1,\dots,\pi_{m_0})$ which can be found by solving $\pi \mathbf{Q} = 0$ with $\sum_{k}\pi_k=1$. We will also assume that the Markov chain is independent of the Brownian motions so that:
	\begin{equation}\label{eq:tran}
		\begin{array}{ll}
			&\disp \PP\{ r(t+\Delta)=j | r(t)=i,
			r(u), s\leq t\}=q_{ij}\Delta + o(\Delta), \text{\quad \quad if } i\ne j\
			\hbox{, and }\\
			&\disp \PP\{ r(t+\Delta)=i \,| r(t)=i,
			r(u), s\leq t\}=1 + q_{ii}\Delta + o(\Delta).
		\end{array}
	\end{equation}
	Intuitively, the above equations tells us that if we know that at $t$ the environment is in state $i$, then the probability of jumping in a small time $\Delta$ to state $j$ is approximately $q_{ij}\Delta$.
	
	For the field of epidemiology, one of the fundamental questions is analyzing when a disease is endemic, i.e., persisting over long periods of time, and when it is epidemic, i.e., when it only lasts a short time.
	Our main goal is to provide sharp conditions which tell us when the disease is endemic and when it is epidemic.
	Furthermore, we are able to characterize endemic diseases by proving that the stochastic process converges to a stationary distribution exponentially fast.
	When the disease is an epidemic, we can prove that the disease disappears exponentially fast and can find the rate of extinction. 
	
	The paper is structured as follows.
	In Section \ref{sec:main}, we present the main results and derive the condition for extinction.
	In Section \ref{s:pers}, we prove when one has the persistence of the disease.
	In Section \ref{sec:disc}, we provide a discussion of the results and simulations.
	
	\section{Main results and extinction} \label{sec:main}
	Let $\big(\Omega, \mathcal F, \{\mathcal F_t\}_{t\geq 0}, \mathbb P\big)$ be a complete probability space with a filtration $\{\mathcal F_t\}_{t\geq 0}$, which satisfies the usual conditions.
	Define $c_1=\mu,$ $c_2=\mu+\gamma_1+\gamma_2$, $c_3=\mu+\gamma_3+\gamma_4$, $c_4=\mu+\gamma_5$, which all depend on $r(t)$, and set $j=(s,i,q,r,k)$, $j_0=(s,0,q,r,k)$, $J(t):=(S(t), I(t), Q(t), R(t), r(t))$. Throughout this paper, we use the lower case letters $s,i,q,r$ to represent the initial values of $S(t), I(t),Q(t), R(t)$ respectively. Denote $\R_+^4:=\{(s,i,q,r)\in \R^4: s\ge 0,i\ge 0,  q\ge 0, r\ge 0 \}$, $\R_+^{4,\circ}:=\{(s,i,q,r)\in \R^4: s> 0,i> 0,  q> 0, r> 0\}$, and $\R_+^{4,\star}:=\{(s,i,q,r)\in \R^4: s\ge 0,i> 0,  q\ge 0, r\ge 0\}$. Moreover, to simplify the notation, we let $\check{a}:=\max_{k\in \mathcal M}\{a(k)\}$, and $\hat{a}:=\min_{k\in \mathcal M}\{a(k)\}$.
	
	The following assumption is supposed to hold throughout the paper.
	\begin{asp}
		$F(s,i,k)$ is a locally Lipschitz function in $(s,i)\in\R^2_+$  and continuous in $i$ at $i=0$ uniformly in $s$, i.e.,
		\begin{equation}\label{e1-a1}
			\lim_{i\to0}\sup_{s\geq 0,k\in\M}\left\{|F(s,i,k)-F(s,0,k)|\right\}=0.
		\end{equation}
		Suppose further that $F(s,0,k)$ is non-decreasing in $s$ for each fixed $k\in\M$ and
		$F(s,0,k)\leq C_F(s+1),\,\forall\, (s,k)\in[0,\infty)\times\M$ for some constant $C_F>0$.
	\end{asp}
	\begin{rem}
		The functional forms $F(S, I) = \beta S/(1 + m_1 S)$ (Holling type II), $F(S, I) = \beta S$ (bilinear functional response), and $F(S, I) = \beta S/(1 + m_1 S + m_2 I)$ (Beddington-DeAngelis functional response) satisfy  Assumption \ref{e1-a1}. Moreover, we can also include functional responses of the form $IS/F(S,I)$, from \cite{nguyen2020long}, by replacing $S/F(S,I)$ with $F(S,I)$. As a result one can include, for example, functional responses of the form $IS/(S+I)$ and of the form $\beta IS/(1+pI+qI^2)$.
	\end{rem}
		\begin{thm}\label{thm:wp} The system \eqref{ww1} has the following properties:
		\begin{itemize}
			\item[{\rm(i)}]
			For any initial value $(s,i,q,r,k)\in\R^4_+\times\M$, there exists a unique global solution
			$\big(S(t), I(t), Q(t), R(t)\big)$ to \eqref{ww1}
			such that $\PP_j\left\{\big(S(t), I(t), Q(t), R(t)\big)\in\R^4_+,\  \forall t\geq0\right\}=1$, where $j$ denotes the initial condition.
			Moreover, $\PP_{j_0}\left\{I(t)=0,\  \forall t\geq0\right\}=1$ and
			\begin{equation}\label{ea-thm2.1}
				\PP_{j}\left\{\big(S(t), I(t), Q(t), R(t)\big)\in\R^{4,\circ}_+,\  \forall t>0\right\}=1,\text{ for any }j\in\R^{4,*}_+\times\M.
			\end{equation}
			In addition, the solution process $\big(S(t), I(t),Q(t),R(t), r(t)\big)$ is a Markov-Feller process with transition probability denoted by $P(t,s,i,q,k,\cdot)$.
			\item[{\rm(ii)}] For any $p>0$ sufficiently small, there exist $C_p>0$ and $D_p>0$ such that
			\begin{equation}\label{e1-thm2.1}
				\E_{j}\big(S(t)+I(t)+Q(t)+R(t)\big)^{1+p}\leq \dfrac{(1+s+i+q+r)^{1+p}}{e^{D_pt}}+\frac{C_p}{D_p}, \quad\forall t\geq 0.
			\end{equation}
			Moreover, for any $H,\eps, T>0$, there exists an $M_{H,\eps,T}>0$ such that
			\begin{equation}\label{e2-thm2.1}
				\PP_{j}\left\{\sup_{ t\in[0,T]}\big\{S(t)+I(t)+Q(t)+R(t)\big\}\leq M_{H,\eps, T}\right\}\geq 1-\eps, \quad \forall j\in[0,H]^4\times\M.
			\end{equation}
		\end{itemize}
	\end{thm}
	
	Consider the dynamics of the population without disease by letting $I(t)=Q(t)=R(t)=0$ in the first equation of \eqref{ww1}. We obtain
	\begin{equation}\label{ww2}
		\begin{aligned}
			d\Se(t)=&\Big[A\big(r(t)\big)-\mu\big(r(t)\big)\Se(t)\Big]dt+\sigma_1\big(r(t)\big)\Se(t)dW_1(t)
		\end{aligned}
	\end{equation}
	
	Lemma \eqref{lm} will show that the process $\big(\Se(t),r(t)\big)$ has a unique invariant probability measure $\bnu_0$ supported on $(0,\infty)\times \mathcal M$. Moreover, the measure $\bnu_0$ can be regarded as the unique invariant measure on the boundary of $\big(S(t),I(t),Q(t),R(t),r(t)\big)$ by embedding $[0,\infty)\times \mathcal M$ to $[0,\infty)\times \{0\}^3\times\mathcal M$.
	Let 
	\begin{equation}\label{e:lambda}
		\lambda=\sum_{k\in\M}\int_{(0,\infty)}\left[F(s,0,k)-c_2(k)-\dfrac{\sigma_2^2(k)}{2}\right]\bnu_0\big(ds,\{k\}\big).
	\end{equation}
	
	\begin{rem}
		In particular, if there is no switching, one has
		\begin{equation}\label{e:lambda_nor}
			\int_{(0,\infty)}\left[F(s,0)-c_2-\dfrac{\sigma_2^2}{2}\right]\bnu_0(dx).
		\end{equation}
		The SDE 	
		\begin{equation}\label{ww2_nor}
			\begin{aligned}
				d\Se(t)=&\Big[A-\mu\Se(t)\Big]dt+\sigma_1\Se(t)dW_1(t)
			\end{aligned}
		\end{equation}
		has a unique stationary distribution which can be found explicitly (see Lemma 2.2 in \cite{HN23}) so that if $F(s,i)=\beta s$ then
		\[
		\lambda =\int_{(0,\infty)}\left[\beta s-c_2-\dfrac{\sigma_2^2}{2}\right]\bnu_0(dx) = \beta \frac{A}{\mu} -c_2 - \frac{\sigma_2^2}{2}.
		\]
		
	\end{rem}
	
	\begin{rem}
		If $F(s,i,k)=\beta(k) s$ and we set $\mathbf{Q} = (q_{kl})_{m_0\times m_0}$, $\mathbf{U} :=\diag (\mu(1),\dots,\mu(m_0))$, $\beta :=(\beta(1),\dots,\beta(m_0))^T$ then one can show \cite{WL23} that
		\[
		\lambda = \sum_{k=1}^{m_0} \pi_k\left(\kappa_kA(k) -c_2(k) - \frac{\sigma_2^2(k)}{2}\right).
		\]
		where $K=(\kappa_1,\dots,\kappa_{m_0})^T$ is the unique solution to
		\[
		(\mathbf{U} - \mathbf{Q})K=\beta.
		\]
		
		Specifically, if $m_0=2$ we get
		\[
		\lambda =  \pi_1\left(\kappa_1A(1) -c_2(1) - \frac{\sigma_2^2(1)}{2}\right)+\pi_2\left(\kappa_2A(2) -c_2(2) - \frac{\sigma_2^2(2)}{2}\right)
		\]
		where 
		\[
		\kappa_1 = \frac{(\mu(2)-q_{22})\beta(1)+q_{12}\beta(2)}{(\mu(1)-q_{11})(\mu(2)-q_{22})-q_{12}q_{21}}, 
		\]
		\[
		\kappa_2 = \frac{(\mu(1)-q_{1})\beta(2)+q_{21}\beta(1)}{(\mu(1)-q_{11})(\mu(2)-q_{22})-q_{12}q_{21}}, 
		\]
		and
		\[
		\pi_1 = \frac{q_{21}}{q_{12}+q_{21}}, \pi_2 = \frac{q_{12}}{q_{12}+q_{21}}.
		\]
	\end{rem}
	
	\begin{rem}
		We note that $\lambda$ from \eqref{e:lambda} would be unchanged if we would ignore the class $Q$ and would just look at the SIR model 
		\begin{equation}\label{e:SIR}
			\begin{cases}
				\begin{aligned}
					dS(t)&=\Big[A\big(r(t)\big)-F\big(S(t),I(t),r(t)\big)I(t)-\mu\big(r(t)\big)S(t)+\gamma_5\big(r(t)\big)R(t)\Big]dt\\
					&+\sigma_1\big(r(t)\big)S(t)dW_1(t),\\
					dI(t)&=\Big[F(S(t),I(t),r(t))-\mu(r(t))-\gamma_1(r(t))-\gamma_2(r(t))\Big]I(t)dt\\
					&+\sigma_2\big(r(t)\big)I(t)dW_2(t),\\
					dR(t)&=\Big[\gamma_2\big(r(t)\big)I(t)-\big(\mu\left(r(t)\right)+\gamma_5\left(r(t)\right)\big)R(t)\Big]dt\\
					&+\sigma_4\big(r(t)\big)R(t)dW_4(t).\\
				\end{aligned}
			\end{cases}
		\end{equation}
		This happens because $\lambda$ measures the `invasion rate' of the infected class $I$ into the susceptible class $S$ at stationarity. Note that \eqref{ww2} looks at the susceptible class $S$ in the absence of any of the other classes so that the stationary measure $\bnu_0$ only depends on $A$, $\mu$, and $\sigma_1$. Similarly, in \eqref{e:lambda} we only get dependence on the coefficients from the SDE governing the infected class $I$. Since we will show that the sign of $\lambda$ determines the persistence or extinction of the diseases, the only influence $Q$ can have will be on the distribution at stationarity of the infected individuals, if the disease persists. 
	\end{rem}	

	\subsection{Extinction.}
	In this subsection we will look at the case $\lambda<0$. 	
	Define the randomized occupation measures
	$$
	\wdt\Pi^t(\cdot):=\dfrac1t\int_0^t\1_{\left\{(S(u),I(u),Q(u),R(u),r(u))\in\cdot\right\}}du, \; t > 0.
	$$
	\begin{thm}\label{t:ext}
		If $\lambda<0$ then for any $j:=(s,i,q,r,k)\in\R^{4,\circ}_+\times\M$ we have
		\begin{equation}\label{e:ext}
			\begin{aligned}
				\PP_{j}\bigg\{\wdt\Pi^t \to \bnu_0~\text{weakly},  \lim_{t\to\infty}\dfrac{\ln I(t)}t=\lambda<0, \lim_{t\to \infty} \dfrac{\ln Q(t)}{t}<0, \lim_{t\to \infty} \dfrac{\ln R(t)}{t}<0\bigg\}=1.
			\end{aligned}
		\end{equation}	
	\end{thm}
	
	We will make use of the following result from \cite{phu2020longtime}.
	
	\begin{lm}\label{lm} \textup{(}Lemma 3.1 in \cite{phu2020longtime}\textup{)}
		For any $\epsilon_0\in[0,1]$, there exists a unique invariant probability measure $\bmu_{\epsilon_0}$ to $(\Se^{\epsilon_0}(t),r(t))$, where $(\Se^{\epsilon_0}(t),r(t))$ is the solution of
		\begin{equation}\label{ww2'}
			d\bar S^{({\epsilon_0})}(t)=\Big[A(r(t))-\mu(r(t))\overline S^{(\epsilon_0)}(t)+\gamma_5(r(t))\delta_0\Big]dt+\sigma_1( r(t))\overline S^{(\epsilon_0)}(t)dW_1(t),
		\end{equation}
		In addition, any measure function $f(s,k)$ satisfying $f(s,k)\leq C_f(s+1)$ is $\bmu_{\epsilon_0}$-integrable.
		Moreover, if $\lambda<0$, there exists $\epsilon_0>0$ such that
		\begin{equation}\label{theta3}
			\wdt\lambda\leq\frac34\lambda,
		\end{equation}
		where
		$$\wdt\lambda:=\sum_{k\in\M}\int_{(0,\infty)}\left[F(s,0,k)-c_2(k)-\dfrac{\sigma_2^2(k)}{2}\right]\bnu_{\epsilon_0}(dx,\{k\}).$$
		and that
		\begin{equation}\label{theta2}
			|F(s,i,k)-F(s,0,k)|\leq \frac{|\lambda|}4,\,\text{ for all }\,0\leq i\leq\epsilon_0, s\geq 0,k\in\M.
		\end{equation}
		Finally, it is noted that when $\epsilon_0=0$, \eqref{ww2'} becomes \eqref{ww2}. As a consequence, the above conclusions hold for the solution of \eqref{ww2}.
	\end{lm}
	The next result tells us that if $\lambda<0$ and we start with small densities of quarantined and recovered individuals then the number of infected, quarantined and recovered individuals tends asymptotically to zero with high probability. 
	\begin{lm}\label{lm1.2}
		For any $\eps > 0$ and  $H> 0$, there exists a   $\delta_1> 0$ such that for all $(s,i,q,r,k)\in [0,H]\times (0,\delta_1]^3\times\M$, we have
		\begin{align}
			\label{eq-lem31}
			\PP_{j}\bigg\{ \lim_{t\to\infty}\dfrac{\ln I(t)}t<0, &\lim_{t\to \infty} \dfrac{\ln Q(t)}{t}<0, \nonumber \lim_{t\to \infty} \dfrac{\ln R(t)}{t}<0 \bigg\} \geq 1-\eps.
		\end{align}
	\end{lm}
	\begin{proof}
		Denote $g(s,i,k):=F(s,i,k)-c_2(k)-\dfrac{\sigma_2^2(k)}2.$
		By the ergodicity of $\big(\bar S^{({\epsilon_0})}(t),r(t)\big)$,
		\begin{equation}
			\label{e-Omega0}
			\lim_{t\to\infty}\dfrac1t\int_0^tg\big(\bar S^{({\epsilon_0})}(u),0,r(u)\big)du=\sum_{k\in\M}\int_{(0,\infty)}g(s,0,k)\bmu_{\epsilon_0}(ds,k)=\wdt\lambda \text{ a.s.}
		\end{equation}
		Therefore, for any $\eps > 0$, there exists a $T_1=T_1(H,\eps) > 0$ such that $\PP_{H,i}(\Omega_{1}) \ge 1-\frac\eps4$, where
		\begin{equation}\label{omega1}\Omega_1=\left\{\omega\in\Omega: \dfrac1t\int_0^tg(\bar S_s^{({\epsilon_0})}(u),0, r(u))du\leq \wdt\lambda+\dfrac{|\lambda|}{4}\quad \text{ for all } t \ge T_1\right\}
		\end{equation}
		where the subscript in $\PP_{H,i}$ indicates the initial value of $(\bar S^{({\epsilon_0})}(u),r(u))$.
		Because of the uniqueness of solutions, we have $\Se^{(\epsilon_0)}_{s,k}(u)\leq\Se^{(\epsilon_0)}_{H,i}(u), \forall u\geq 0$ almost surely where the subscript of $\Se^{(\epsilon_0)}_{s,k}$ indicates the initial value $(\Se^{(\epsilon_0)}(0),r(0))$, which implies $\PP_{s,k}(\Omega_1)\geq 1-\frac\eps4$ for $(s,k)\in[0,H]\times\M$.
		
		Let  $\epsilon_1$ be a positive constant such that $0<\epsilon_1<\min\{\frac{-\lambda}{8}, \hat{c}_3+\frac{{\hat\sigma}_3^2}{2}, {\hat c}_4+\frac{{\hat\sigma_4}^2}{2}\}$.
		
		From the strong law of large numbers for Brownian motions we get
		\begin{equation}
			\label{e-slln-bm}
			\lim_{t\to\infty}\dfrac1t\int_0^t\sigma_k(r(u))dW_k(u)=0, \text{ a.s. for } k=1,2,3,4.
		\end{equation}
		
		This means there is a $T_2(\eps)>0$ such that  $\PP(\Omega_2)\geq 1- \frac\eps4, $ where  	\begin{equation}\label{omega2}\Omega_2=\left\{\omega\in\Omega: \dfrac1t\left|\int_0^t\sigma_k(r(u))dW_k(u)\right|\leq \min\left\{\dfrac{|\lambda|}{4},\epsilon_1\right\} \quad \text{ for all } t \ge T_2, k=1,2,3,4 \right\}.
		\end{equation}
		Let $T=\max\{T_1,T_2\}$. As a consequence of \eqref{e2-thm2.1} and the continuity of $F(s,i,k)$, there exists $M:=M(\epsilon, T, H)>0$ such that
		$
		\PP(\Omega_3)\geq 1-\frac\eps4, \text{ for all }(s,i,q,r,k)\in(0,H]^4\times\M,
		$
		where
		\begin{equation}\label{Omega4}
			\Omega_3=\left\{\omega\in\Omega: \int_0^TF(\Se^{(\epsilon_0)}(u),0,r(u))du\leq \frac{M}2\right\}.
		\end{equation}
		Moreover, we can choose $M$ large enough so that $\PP_{s,i,k}(\Omega_4)\geq 1-\frac\eps4$, where
		\begin{equation}\label{Omega3}
			\Omega_4=\left\{\omega\in\Omega: \left|\int_0^t\sigma_k(r(u))dW_k(u)\right|\leq \frac{M}2,
			\text{ for all } t\in[0,T] \text{ and } k=1,2,3,4\right\}.
		\end{equation}
		Let $\delta_1>0$ such that
                \begin{align*}
                    \delta_1\Bigg(&1+e^{M}+e^{M/2}+\check{\gamma}_2e^{2M}T+e^{M/2}  {\check\gamma}_3( e^{M/2}+ {\check\gamma}_2e^{2M}T) e^{M/2}T\\
                                  &+\dfrac{\check{\gamma}_2}{\hat{c}_3+\frac{\hat{\sigma}_3}{2}+\epsilon_1+\frac{\lambda}4}+\dfrac{{\check\gamma}_3\delta_1{\check\gamma}_2}{\left(\frac \lambda 4+{\hat c}_3+\frac{{\hat\sigma_4}^2}{2}+\epsilon_1\right)\left(\frac \lambda 4+{\hat c}_4+\frac{{\hat\sigma_4}^2}{2}+3\epsilon_1\right)}\Bigg)\\
                                  &<\epsilon_0.
                \end{align*}
	
		From the second equation of \eqref{ww1} we get
		\begin{equation}\label{ii1}			
	I(t)=I(0)\exp\left\{\int_0^t \left(F(S(u),I(u),r(u))-c_2(r(u))-\dfrac{\sigma_2^2(r(u))}2\right)du+\int_0^t\sigma_2(r(u))dW_2(u)\right\}.
		\end{equation}
		Using \eqref{Omega4} and \eqref{Omega3} we obtain
		\begin{equation}\label{delta_0}
			\begin{aligned}
				I(t)\leq&I(0)\exp\left\{\int_0^tF(S(u), I(u),r(u))du+\int_0^t\sigma_2(r(u))dW_2(u)\right\}\\
				\leq& \delta_1 e^{M}<\epsilon_0 \text{ for any } t\in[0,T], I(0)\leq\delta_1 \text{ and } \omega\in \Omega_3\cap \Omega_4.
			\end{aligned}
		\end{equation}
	Define $$\Phi_t:=\exp\{-\int_0^t \left(c_3(r(s))+\dfrac{{\sigma_3}^2(r(s))}{2}\right)ds+\int_0^t \sigma_3(r(s))dW_3(s)\}.$$
		The third equation of \eqref{ww1} implies that
	$$Q(t)=\Phi_t\left(Q_0+\int_0^t\gamma_2(r(s))I_s\Phi_s^{-1}ds\right)$$

		For $t\in [0, T]$ and $\omega\in \Omega_3\cap\Omega_4$,
		\begin{align*}
			e^{-\int_0^t \left(c_3(r(s))+\frac{{\sigma_3}^2(r(s))}{2}\right)dr-\frac M2}\le \Phi_t\le e^{-\int_0^t \left(c_3(r(s))+\frac{{\sigma_3}^2(r(s))}{2}\right)dr+\frac M2}.
		\end{align*}
	Hence, when $(i,q)\in (0, \delta_1]^2$,
\begin{equation}\label{q1}
		\begin{aligned}
			Q(t)&\le e^{M/2}(\delta_1+\int_0^t {\check\gamma}_2 \delta_1 e^M e^{M/2}ds)\\
			&\le e^{M/2}\delta_1+e^{M/2}  {\check\gamma}_2 \delta_1 e^M e^{M/2}T\\
			&< \epsilon_0.
		\end{aligned}
	\end{equation}
		Now, from the forth equation of \eqref{ww1} we obtain
			$$R(t)=\Psi_t\left(R_0+\int_0^t\gamma_3(r(s))Q_s\Psi_s^{-1}ds\right),$$
			where 
				$$\Psi_t:=\exp\{-\int_0^t \left(c_4(r(s))+\dfrac{{\sigma_4}^2(r(s))}{2}\right)ds+\int_0^t \sigma_4(r(s))dW_4(s)\}.$$
				Similarly, for $t\in [0, T]$ and $\omega\in \Omega_3\cap \Omega_4$ we obtain the following estimation for $\Psi_t$
					\begin{align*}
					e^{-\int_0^t \left(c_4(r(s))+\frac{{\sigma_4}^2(r(s))}{2}\right)dr-\frac M2}\le \Psi_t\le e^{-\int_0^t \left(c_4(r(s))+\frac{{\sigma_4}^2(r(s))}{2}\right)dr+\frac M2}.
				\end{align*}
		This implies that for $(i,q,r)\in (0, \delta_1]^3$,
		\begin{equation}\label{r1}
		\begin{aligned}
			R(t)
			&\le e^{M/2}\delta_1+e^{M/2}  {\check\gamma}_3( e^{M/2}\delta_1+e^{M/2}  {\check\gamma}_2 \delta_1 e^M e^{M/2}T) e^{M/2}T\\
			&< \epsilon_0.
		\end{aligned}
	\end{equation}

		Now, we define the stopping time
		\begin{equation}\label{wtau}
			\wtau: =\inf\left\{t\geq 0:  I(t)\vee Q(t)\vee R(t)\geq \epsilon_0\right\}.
		\end{equation}
		As a consequence, for $\omega\in\Omega_3\cap\Omega_4$ we have $\wtau>T$.
		
		Note that, one has
		$S(t)\leq\Se^{(\epsilon_0)}(t)$ for $t\in [T, \wtau]$ given that they have the same initial value.
                Then for almost every $\omega\in\bigcap_{j=1}^4\Omega_j$, and $i\in (0, \delta_1]$,

		\begin{equation}\label{tau1}
		\begin{aligned}
			I(t)&=	I(0)\exp\Big\{\int_0^t \left(F(S(u),0,r(u))-c_2(k)-\frac{\sigma_2^2(r(u))}2\right)du+\int_0^t\sigma_2(r(u))dW_2(u)\Big\}\\
			&\leq I(0)\exp\left\{\int_0^tg(\Se^{(\epsilon_0)}(u),0,r(u))du+\int_0^t\sigma_2(r(u))dW_2(u)\right\}\\
			&\leq I(0)\exp\left\{\wdt\lambda t+\frac{|\lambda| t}4+\frac{|\lambda| t}4\right\}\\
                        &\leq \delta_1\exp\left\{t\frac{\lambda}4\right\}<\epsilon_0, \text{ for all }t\in[T,\wtau).
		\end{aligned}
		\end{equation}
		Furthermore, 
		\begin{align*}
			Q(t)&=\Phi_t\left(Q_0+\int_0^t\gamma_2(r(s))I_s\Phi_s^{-1}ds\right)\\
			&=\Phi_t\left(Q_0+\int_0^T\gamma_2(r(s))I_s\Phi_s^{-1}ds\right)+\Phi_t \int_T^t \gamma_2(r(s))I_s\Phi_s^{-1}ds.
		\end{align*}
	Note that for all $\omega\in \cap_{j=1}^4 \Omega_j$,
	\begin{equation}\label{q2}
	\begin{aligned}
	 \Phi_t\int_T^t \gamma_2(r(s))I_s\Phi_s^{-1}ds&=\int_T^t  \gamma_2(r(s))I_s e^{-\int_s^t \left(c_3(r(u))+\frac{{\sigma_3}^2(r(u))}{2}\right)du+\int_s^t \sigma_3(r(u))dW_3(u)}ds\\
	 &\leq \int_T^t {\check\gamma}_2\delta_1 e^{\frac{s\lambda}4}e^{-\int_s^t \left({\hat c}_3+\frac{{\hat\sigma_3}^2}{2}\right)du+\epsilon_1(t+s)}ds\\
	 &=\int_T^t {\check\gamma}_2\delta_1 e^{\frac{s\lambda}4}e^{-\left({\hat c}_3+\frac{{\hat\sigma_3}^2}{2}\right)(t-s)+\epsilon_1(t+s)}ds\\
	 &={\check\gamma}_2\delta_1 e^{(-{\hat c}_3-\frac{{\hat\sigma_3}^2}{2}+\epsilon_1)t}\int_T^t e^{(\frac \lambda 4+{\hat c}_3+\frac{{\hat\sigma_3}^2}{2}+\epsilon_1)s}ds\\
	 &\le {\check\gamma}_2\delta_1 e^{(-{\hat c}_3-\frac{{\hat\sigma_3}^2}{2}+\epsilon_1)t} \dfrac{e^{(\frac \lambda 4+{\hat c}_3+\frac{{\hat\sigma_3}^2}{2}+\epsilon_1)t}}{\frac \lambda 4+{\hat c}_3+\frac{{\hat\sigma_3}^2}{2}+\epsilon_1}\\
	 &=\dfrac{{\check\gamma}_2\delta_1}{\frac \lambda 4+{\hat c}_3+\frac{{\hat\sigma_3}^2}{2}+\epsilon_1}e^{(\frac{\lambda}{4}+2\epsilon_1)t}\\
	 &\le \dfrac{{\check\gamma}_2\delta_1}{\frac \lambda 4+{\hat c}_3+\frac{{\hat\sigma_3}^2}{2}+\epsilon_1}.
	\end{aligned}
\end{equation}
Let $n_1>T$ be an positive integer. For $t\in [0, \wtau\wedge n_1]$, by \eqref{q1} and \eqref{q2} we obtain
\begin{equation}\label{tau2}
	Q_t\le \delta_1\left(e^{M/2}+e^{2M}  {\check\gamma}_2T+\dfrac{{\check\gamma}_2}{\frac \lambda 4+{\hat c}_3+\frac{{\hat\sigma_3}^2}{2}+\epsilon_1}\right)<\epsilon_0.
\end{equation}
Similarly,
$$R(t)=\Psi_t\left(R_0+\int_0^T\gamma_3(r(s))Q_s\Psi_s^{-1}ds\right)+\Psi_t\int_T^t \gamma_3(r(s))Q_s \Psi_s^{-1}ds.$$
For the last term, 
\begin{equation}\label{r2}
	\begin{aligned}
		\Psi_t\int_T^t \gamma_3(r(s))Q_s\Psi_s^{-1}ds&=\int_T^t  \gamma_3(r(s))Q_s e^{-\int_s^t \left(c_4(r(u))+\frac{{\sigma_4}^2(r(u))}{2}\right)du+\int_s^t \sigma_4(r(u))dW_4(u)}ds\\
		&\leq \int_T^t {\check\gamma}_3\delta_1 \dfrac{{\check\gamma}_2}{\frac \lambda 4+{\hat c}_3+\frac{{\hat\sigma_3}^2}{2}+\epsilon_1}e^{(\frac{\lambda}{4}+2\epsilon_1)s}e^{-\int_s^t \left({\hat c}_4+\frac{{\hat\sigma_4}^2}{2}\right)du+\epsilon_1(t+s)}ds\\
		&=\int_T^t {\check\gamma}_3\delta_1 \dfrac{{\check\gamma}_2}{\frac \lambda 4+{\hat c}_3+\frac{{\hat\sigma_3}^2}{2}+\epsilon_1}e^{(\frac{\lambda}{4}+2\epsilon_1)s}e^{-\left({\hat c}_4+\frac{{\hat\sigma_4}^2}{2}\right)(t-s)+\epsilon_1(t+s)}ds\\
		&={\check\gamma}_3\delta_1 \dfrac{{\check\gamma}_2}{\frac \lambda 4+{\hat c}_3+\frac{{\hat\sigma_3}^2}{2}+\epsilon_1}e^{(-{\hat c}_4-\frac{{\hat\sigma_4}^2}{2}+\epsilon_1)t}\int_T^t e^{(\frac \lambda 4+{\hat c}_4+\frac{{\hat\sigma_4}^2}{2}+3\epsilon_1)s}ds\\
		&\le {\check\gamma}_3\delta_1\dfrac{{\check\gamma}_2}{\frac \lambda 4+{\hat c}_3+\frac{{\hat\sigma_3}^2}{2}+\epsilon_1} e^{(-{\hat c}_4-\frac{{\hat\sigma_4}^2}{2}+\epsilon_1)t} \dfrac{e^{(\frac \lambda 4+{\hat c}_4+\frac{{\hat\sigma_4}^2}{2}+3\epsilon_1)t}}{\frac \lambda 4+{\hat c}_4+\frac{{\hat\sigma_4}^2}{2}+3\epsilon_1}\\
		&=\dfrac{{\check\gamma}_3\delta_1{\check\gamma}_2}{\left(\frac \lambda 4+{\hat c}_3+\frac{{\hat\sigma_4}^2}{2}+\epsilon_1\right)\left(\frac \lambda 4+{\hat c}_4+\frac{{\hat\sigma_4}^2}{2}+3\epsilon_1\right)}e^{(\frac{\lambda}{4}+4\epsilon_1)t}\\
		&\le \dfrac{{\check\gamma}_3\delta_1{\check\gamma}_2}{\left(\frac \lambda 4+{\hat c}_3+\frac{{\hat\sigma_4}^2}{2}+\epsilon_1\right)\left(\frac \lambda 4+{\hat c}_4+\frac{{\hat\sigma_4}^2}{2}+3\epsilon_1\right)}
	\end{aligned}
\end{equation}
Let $n_2>T$ be an positive integer, using \eqref{r1} and \eqref{r2} yields, 
\begin{equation}\label{tau3}
	R_t\le \delta_1\left(e^{M/2}+e^{M/2}{\check \gamma}_3(e^{M/2}+e^{2M}{\check\gamma}_2T)e^{M/2}T+\dfrac{{\check\gamma}_3\delta_1{\check\gamma}_2}{\left(\frac \lambda 4+{\hat c}_3+\frac{{\hat\sigma_4}^2}{2}+\epsilon_1\right)\left(\frac \lambda 4+{\hat c}_4+\frac{{\hat\sigma_4}^2}{2}+3\epsilon_1\right)}\right)<\epsilon_0,
\end{equation}
for $t\in [0, \wtau\wedge n_2]$.

By the definition of $\wtau$, and combining \eqref{tau1}, \eqref{tau2}, \eqref{tau3} we have $\wtau>\max\{n_1,n_2\}$. Note that $n_1, n_2$ are arbitrary, hence $\wtau=\infty$ for all $\omega\in\bigcap_{j=1}^4\Omega_j$, and $(I(0), Q(0), R(0))\in (0,\delta_1]^3$.

Therefore,

 $\lim_{t\to \infty} \dfrac{\ln I(t)}{t}\le \dfrac{\lambda}{4}<0,$ $\lim_{t\to \infty} \dfrac{\ln Q(t)}{t}\le \dfrac\lambda 4+2\epsilon_1<0$, and $\lim_{t\to \infty} \dfrac{\ln R(t)}{t}\le\dfrac\lambda 4+4\epsilon_1<0$.

The proof of the lemma is therefore complete by noting that $P(\cap_{j=1}^4 \Omega_j)\ge 1-\epsilon$.

	\end{proof}
	
	\begin{proof}[Proof of Theorem \ref{t:ext}]
		With probability 1, any weak-limit (if it exists) of $\wdt\Pi^t(\cdot)$ as $t\to\infty$
		is an invariant probability measure of the process $(S(t),I(t),Q(t),R(t),r(t))$ on $\R^4_+\times\M$.
		From the Lemma \ref{lm1.2} we see that the collection of  measures $\{\wdt\Pi^t(\cdot;\omega), t> 0, \omega\in \cap_{j=1}^{4} \Omega_{j}\}$ is tight in $\R^4_+\times\M$ and any weak limit of $\wdt\Pi^t(\cdot)$ as $t\to\infty$ must have support on $[0,\infty)\times\{0\}\times\{0\}\times\{0\}\times\M.$
		Furthermore, $\bmu_0$ -- {\it regarded as an invariant measure of $\big(S(t), I(t), Q(t), R(t), r(t)\big)$} -- 
		is the unique invariant probability measure on $[0,\infty)\times\{0\}\times\{0\}\times\{0\}\times\M$.
		Therefore,
		$\wdt\Pi^t(\cdot)$ converges weakly to $\bmu_0$ for almost every $\omega\in \cap_{j=1}^4\Omega_j$ as $t\to \infty$.
	This means that
		\begin{equation}
			\label{e-s/F(s,I)-conv}
			\begin{aligned}
				\lim_{t\to\infty}&\dfrac1t\int_0^t g(S(u),I(t),r(u))du=\sum_{k\in\M}\int_{(0,\infty)}g(s,0, k)\bmu_0(ds,k)=\lambda,
			\end{aligned}
		\end{equation}
		for almost every $\omega\in \bigcap_{j=1}^4\Omega_j$.
		
		In addition, note that
		$$\limsup_{t\to\infty}\dfrac1t\int_0^t (S(u))^{1+\wdt p} du\leq \lim_{t\to\infty}\dfrac1t\int_0^t (\Se^{(\delta_0)}(u))^{1+\wdt p}du=\sum_{k\in\M}\int_{(0,\infty)}s^{1+\wdt p}\bmu_{\delta_0}(ds,k)<\infty,$$	
		for some $\wdt p>0$. Therefore, in light of \cite[Lemma 5.6]{hening2018coexistence},  the limit \eqref{e-s/F(s,I)-conv} is valid.
		
		Using the definition of $g(s,i,k)$ and \eqref{ii1}, we obtain
		$$\dfrac{\ln I(t)}t=\dfrac{\ln I(0)}t+\dfrac1t\int_0^t g(S(u),I(t),r(u))du+\dfrac1t\int_0^t\sigma_2(r(u))dW_2(u).$$
		By letting $t\to\infty$ and using  \eqref{e-slln-bm} and \eqref{e-s/F(s,I)-conv}, we obtain that for almost every $\omega\in \bigcap_{j=1}^4\Omega_j$,
		$
		\lim_{t\to\infty}\frac{\ln I(t)}t=\lambda$.
		
		In view of Lemma \ref{lm1.2},
		the process $(S(t), I(t), Q(t), R(t))$ is transient in $\R^{4,\circ}_+$.
		As a result the process has no invariant probability measure in $\R^{4,\circ}_+$ and $\bmu_0$ is the unique invariant probability measure of $(S(t), I(t) , Q(t), R(t), r(t))$   in $\R^4_+\times\M$.
		Let $H$ be sufficiently large so that $\bmu_0(\{s\in(0,H)\})>1-\frac{\eps}2$.
		Due to \eqref{e1-thm2.1}, the process  $(S(t), I(t), Q(t), R(t), r(t))$ is tight.
		This implies that the occupation measure
		$$\Pi^t_{j}(\cdot):=\dfrac1t\int_0^t\PP_{j}\left\{(S(u),I(u),Q(u),R(u),r(u))\in\cdot\right\}du$$
		is tight in $\R^4_+\times\M$.
		By using the fact that any weak-limit of $\Pi^t_{j}$ as $t\to\infty$ must be an invariant probability measure of $(S(t),I(t),Q(t), R(t),r(t))$,
		we have that
		$\Pi^t_{j}$ converges weakly to $\bmu_0$ as $t\to\infty$.
		Thus, for any $\delta>0$,
		$$
		\liminf_{t\to\infty}\Pi^t_{s,i,k}((0,H)\times(0,\delta)\times (0,\delta)\times (0,\delta)\times\M)\geq\bmu_0((0,H)\times(0,\delta)\times (0,\delta)\times (0,\delta)\times\M)>1-\frac{\eps}2.
		$$
		Thus, there exists a  positive $\hat T$ such that
		$$\Pi^{\hat T}_{j}((0,H)\times(0,\delta)\times (0,\delta)\times (0,\delta)\times\M)>1-\eps,$$
		i.e.,
                $$\frac{1}{\hat T}\int_0^{\hat T}\PP_{s,i,k}\{(S(t),I(t),Q(t), R(t))\in (0,H)\times(0,\delta)^3\}dt>1-\eps.$$
		Let $\hat\tau=\inf\{t\geq 0: (S(t),I(t), Q(t), R(t))\in (0,H)\times(0,\delta)^3\}$, then
		$$\PP_{j}\{\hat\tau\leq\hat T\}>1-\eps.$$
	
		By the strong Markov property and Lemma \ref{lm1.2}, we have that
		\begin{align*}
			\PP_{j}\bigg\{\lim_{t\to\infty}\dfrac{\ln I(t)}t=\lambda<0, \lim_{t\to \infty} \dfrac{\ln Q(t)}{t}<0, \lim_{t\to \infty} \dfrac{\ln R(t)}{t}<0\bigg\}\geq 1-\eps.
		\end{align*}
		for any $(s,i,q,r,k)\in\R^{4,*}_+\times\M$. Since $\eps >0$ is arbitrary, \eqref{e:ext} follows. This finishes the proof.
	\end{proof}

	\section{Persistence}\label{s:pers} 
	
	In this section we look at the setting when $\lambda>0$. The following lemma will be relevant in the proof of the main result.
	\begin{lm}\textup{(}\cite{hening2018coexistence}, Lemma 3.5\textup{)}
			Let $Y$ be a random variable, $\theta_0>0$ a constant and suppose
			$$\mathbb E\exp(\theta_0 Y)+\mathbb E\exp(-\theta_0Y)\le K_1.$$
			Then the log-Laplace transform $\phi(\theta)=\ln \mathbb E\exp(\theta Y)$ is twice differentiable on $\left[0, \frac{\theta_0}{2}\right)$ and
			$$\dfrac{d\phi}{d\theta}(0)=\mathbb EY,\quad 0\le \dfrac{d^2\phi}{d\theta^2}(0)\le K_2,\quad \theta\in \left[0, \dfrac{\theta_0}2\right),$$
			for some $K_2>0$ depending only on $K_1$.
			
	\end{lm}
	
	\begin{thm}\label{t:pers}
		If $\lambda>0$ then there exists a unique invariant measure $\bmu^*$ of $(S(t), I(t), Q(t), R(t), r(t))$ on $\R^{4,\circ}_+\times\M$. Moreover, the rate of convergence is exponential: there exists $\rho>0$ such that for any $j:=(s,i,q,r,k)\in\R^{4,\circ}_+\times\M$ we have
		$$
		\lim_{t\to\infty} e^{\rho t}\|P(t,j, \cdot)-\bmu^*(\cdot)\|=0
		$$
		where $\|\cdot\| $ is the total variation norm.
	\end{thm}
	\begin{proof}
		The existence of the invariant measure $\bmu^*$ follows by arguments similar to those from the proof of Theorem 2.2 from \cite{nguyen2020general}. As such, we only provide the details of the proof that the rate of convergence is exponential. 
		
		Let $g(s,i,k):=F(s,i,k)-\mu(k)-\nu_1(k)-\nu_2(k)-\dfrac{\sigma_2^2(k)}2.$
		
		Denote by $\pi_r$ the invariant measure of $r(t)$. Since $g(s,0,k)$ is increasing and
			$$
			\lambda=\sum_{k\in\M}\int_{(0,\infty)}\left[g(s,0,k)\right]\bnu_0(ds,\{k\})>0,
			$$
			we obtain
			$$\sum_{k\in \mathcal M} \left(g(s,0,k)\right)\pi_{r}(k)>0.$$
		
		
		
		Thus there exists $s_0>0$ such that
		$$\sum_{k\in \mathcal M} \left(\inf_{s\geq s_0} g(s,0,k)\right)\pi_{r}(k)>0.$$
		Moreover, there exists $\epsilon_1>0$ such that
		$$\lambda'=\sum_{k\in \mathcal M} h_k \pi_{r}(k)>0,$$ where $h_k:=\inf_{(s,i)\in [s_0,\infty)\times [0,\epsilon_1]} g(s,0,k)$.
		
		By Assumption \eqref{e1-a1} and an application of the Fredholm alternative, there exists $\gamma_k>0$ such that
		$$\sum_{j\in \mathcal M} q_{kj}\gamma_j=\lambda'-h_k, \quad \text{for any} \quad k\in \mathcal M.$$
		Now choose $\rho_1$ small enough such that
		$$\rho_1\gamma_k(-\lambda'+h_k)<\lambda'(1-\rho_1\gamma_k), \quad \rho_1\sigma_2^2(k)<\lambda' \quad \text{ and } \rho_1\gamma_k<1,$$
		for any $k\in \mathcal M.$
		
		Define $V_3(j)=(1-\rho_1\gamma_k)i^{-\rho_1}$ and note that
		\begin{equation}\label{ito1}
			\begin{aligned}
				\Lom V_3(j)&=\rho_1 V_3(j)\left(-F(s,i,k)+c_2(k)+\dfrac{\rho_1+1}2 \sigma_2^2(k)\right)\\&+\rho_1(-\lambda'+h_k)V_3(j)+\rho_1^2 \gamma_k(-\lambda'+h_k)i^{-\rho_1}\\&\leq -\dfrac{\rho_1 \lambda' V_3(j)}4,
			\end{aligned}
		\end{equation}
		for any $k\in \mathcal M$, $s\geq s_0$ and $i\leq \epsilon_1.$
		
		Now, since $$[\Lom \frac 1i](j)\leq \dfrac{c_2(k)+\sigma_2^2(k)}i\leq \dfrac{\check c_2+\check \sigma_2^2}i$$ we obtain
		$$\mathbb E_{j}I^{-1}(t)\leq \exp\{(\check c_2+\check\sigma_2^2)t\} i^{-1},$$
		for any $t\ge 0, j\in \mathbb R_+^{4,\circ}\times \mathcal M$.
		
		Since $\lambda=\sum_{k\in\M}\int_{(0,\infty)}\left[F(s,0,k)-c_2(k)-\dfrac{\sigma_2^2(k)}{2}\right]\bnu_0(ds,\{k\})$, we see that there exists $T_1>1$ satisfying 
		\begin{equation}\label{est1}
			-\ln (1-\rho_1\gamma_k)<\dfrac{\rho_1 \lambda T_1}{4}, \quad \forall k\in \mathcal M,
		\end{equation}such that
		$$\dfrac 1T \mathbb E_{j_0}\int_0^T \left(F(s,i,k)-c_2(k)-\dfrac{\sigma_2^2(k)}{2}\right)\,dt>\dfrac {3\lambda}4,$$
		for any $k\in \mathcal M$, $T\geq T_1,$ $s\leq s_0$. Choose $n_{\star}\in \Z_+$ such that $$n_\star>\dfrac{4(\check c_2+\check \sigma_2^2)}{\lambda'}+1.$$
 From Ito's formula,
			$$\ln I(t)=\ln I(0)+G(t),$$
			for some $t\in [T_1,n_{\star}T_1]$. 		
			Using Lemma 3.1 from \cite{hening2018coexistence} and applying the log-Laplace transform for $G(t)$ one can see that there exists $\hat K\ge 0$ such that $$0\le \dfrac{d^2 \hat \phi_{t}}{d\rho_2^2}(\rho_2)\le \hat K\quad \text{ for all } \rho_2\in (0, \rho_1), t\in [T_1, n_{\star}T_1],$$
			where $\hat \phi_{t}(\rho_2)=\ln \mathbb E_j \exp(-\rho_2 G(t))$.

			Similarly to the proof of  Proposition 4.1 from \cite{hening2018coexistence} one can show that there exist $\rho_2\in (0,\rho_1)$ and $\epsilon_2\in (0,\epsilon_1)$ such that
			\begin{equation}\label{est2}
				\mathbb E_{j}I^{-\rho_2}(t)\leq e^{\frac{-\rho_2 \lambda t}2} i^{-\rho_2}, 
			\end{equation}
			for any $k\in \mathcal M$, $t\in [T_1,n_{\star}T_1]$, $s\leq s_0$ and $i<\epsilon_2$.

			%

		Therefore, by \eqref{est1} and \eqref{est2} we have the following estimation
		\begin{equation}\label{s2}
			\begin{aligned}
				\mathbb E_j V_3^{\frac {\rho_2}{\rho_1}}(J(t), r(t))&\leq \mathbb E_j I^{-\rho_2}(t)\\&\leq e^{\frac{-\rho_2 \lambda t}2} i^{-\rho_2}\\&=e^{\frac{-\rho_2 \lambda t}2} V_3^{\frac {\rho_2}{\rho_1}}(j)(1-\rho_1 \gamma_k)^{-\frac {\rho_2}{\rho_1}}\\&=e^{\frac{-\rho_2 \lambda t}2} V_3^{\frac {\rho_2}{\rho_1}}(j)\exp\{\dfrac{\rho_2}{\rho_1} \left(-\ln(1-\rho_1\gamma_k)\right)\}\\&\leq e^{\frac{-\rho_2 \lambda t}2} V_3^{\frac {\rho_2}{\rho_1}}(j)e^{\frac{\rho_2 \lambda T_1}{4} }\\&\leq e^{\frac{-\rho_2 \lambda t}2} V_3^{\frac {\rho_2}{\rho_1}}(j)e^{\frac{\rho_2 \lambda t}{4} }\\&\leq e^{\frac{-\rho_2 \lambda t}4} V_3^{\frac {\rho_2}{\rho_1}}(j),
			\end{aligned}
		\end{equation}
		for any $k\in \mathcal M$, $t\in [T_1, n_\star T_1]$, $s\leq s_0$ and $i<\epsilon_2.$
		
		Since $\rho_2<\rho_1$ one can see from \eqref{ito1} that
		\begin{equation}\label{ito2}
			\begin{aligned}
				\Lom V_3^{\frac{\rho_2}{\rho_1}}(j)\leq -\dfrac{\rho_2 \lambda' V_3^{\frac{\rho_2}{\rho_1}}(j)}4, \quad \text{ for any } k\in \mathcal M, s\ge s_0, i\le \epsilon_2.
			\end{aligned}
		\end{equation}
		Since $g$ is bounded below, $\Lom V_3(j)\leq C V_3(j)$, which implies
			$$\mathbb E_j V_3(J(t))\leq \bar K, t\leq n_*T_1, i\geq \eps_2$$
			for some $\bar K$ which depends on $\eps_2$ and $T_1,n_*$. Next, let $\tau=\inf\{t>0: S(t)\leq s_0 \}$ and $\xi=\inf\{t>0: I(t)>\eps_2\}$.
		 Dynkin's formula yields
			\begin{align*}
				\mathbb E_j[e^{\rho_2\lambda'(\tau\wedge \xi\wedge n_\star T_1)}V_3^{\rho_2/\rho_1}(J(\tau\wedge\xi\wedge n_\star T_1))]&
				\le V_3^{\frac{\rho_2}{\rho_1}}(j)+\mathbb E_j\int_0^{\tau\wedge\xi\wedge n_\star T_1} e^{\rho_2\lambda's}\left(\Lom V_3^{\frac{\rho_2}{\rho_1}}(j(s))+\rho_2\lambda'V_3^{\frac{\rho_2}{\rho_1}}(J(s))\,ds\right)\\&\le V_3^{\frac{\rho_2}{\rho_1}}(j).
			\end{align*}
			Hence,
			\begin{align*}
				V_3^{\frac{\rho_2}{\rho_1}}(j)&\ge \mathbb E_j[\exp(\rho_2\lambda'(\tau\wedge \xi \wedge n_\star T_1))V_3^{\rho_2/\rho_1}(J(\tau\wedge\xi\wedge n_\star T_1))]\\&\ge\mathbb E_j\left[\1_{\{\tau\wedge\xi\wedge (n_*-1)T_1=\tau\}}V_3^{\frac{\rho_2}{\rho_1}}(J(\tau)) \right]+\mathbb E_j\left[\1_{\{\tau\wedge\xi\wedge (n_*-1)T_1=\xi\}}V_3^{\frac{\rho_2}{\rho_1}}(J(\xi)) \right]\\&+\exp(\rho_2 \lambda'(n_*-1)T_1)\mathbb E_j\left[\1_{\{(n_*-1)T_1<\tau\wedge\xi< n_*T_1\}}V_3^{\frac{\rho_2}{\rho_1}}(J(\tau\wedge\xi)) \right]\\&+\exp(\rho_2\lambda'n_*T_1)\mathbb E_j\left[\1_{\{(n_*-1)T_1\geq n_*T_1\}}V_3^{\frac{\rho_2}{\rho_1}}(J(n_*T_1)) \right].
			\end{align*}
			Using the strong Markov property and \eqref{s2} we obtain the following estimates
			\begin{equation}
			\begin{aligned}
				\mathbb E_j \left[\1_{\{\tau\wedge\xi\wedge (n_*-1)T_1=\tau\}}V_3^{\frac{\rho_2}{\rho_1}}(J(n_\star T_1)) \right]&\leq \mathbb E_j \left[\1_{\{\tau\wedge\xi\wedge (n_*-1)T_1=\tau\}}e^{\frac{-\rho_2 \lambda (n_\star T_1-\tau)}4} V_3^{\frac {\rho_2}{\rho_1}}(J(\tau)) \right] \\& \leq \mathbb E_j \left[\1_{\{\tau\wedge\xi\wedge (n_*-1)T_1=\tau\}}V_3^{\frac {\rho_2}{\rho_1}}(J(\tau)) \right].
			\end{aligned}
		\end{equation}
			Furthermore,
			\begin{align}
				\mathbb E_j\left[\1_{\{(n_*-1)T_1<\tau\wedge\xi< n_*T_1\}}V_3^{\frac{\rho_2}{\rho_1}}(J(n_\star T_1)) \right]&\leq \bar{K}\cdot \mathbb E_j\left[\1_{\{(n_*-1)T_1<\tau\wedge\xi< n_*T_1\}}V_3^{\frac{\rho_2}{\rho_1}}(J(\tau \wedge \xi))\right].
			\end{align}
			Therefore,
			\begin{align*}
				V_3^{\frac{\rho_2}{\rho_1}}(j)&\ge\mathbb E_j\left[\1_{\{\tau\wedge\xi\wedge (n_*-1)T_1=\tau\}}V_3^{\frac{\rho_2}{\rho_1}}(J(n_\star T_1)) \right]+\mathbb E_j\left[\1_{\{\tau\wedge\xi\wedge (n_*-1)T_1=\xi\}}V_3^{\frac{\rho_2}{\rho_1}}(J(n_\star T_1)) \right]\\&+(\bar K)^{-1}\exp(\rho_2 \lambda'(n_*-1)T_1)\mathbb E_j\left[\1_{\{(n_*-1)T_1<\tau\wedge\xi< n_*T_1\}}V_3^{\frac{\rho_2}{\rho_1}}(J(n_\star T_1)) \right]\\&+\exp(\rho_2\lambda'n_*T_1)\mathbb E_j\left[\1_{\{(n_*-1)T_1\geq n_*T_1\}}V_3^{\frac{\rho_2}{\rho_1}}(J(n_*T_1)) \right]\\&\ge M  \mathbb E_j V_3^{\frac{\rho_2}{\rho_1}}(J(n_*T_1)),
			\end{align*}
			where $M\ge 2$. Hence, $\E_j V_3^{p_2/p_1}(n_*T_1)\leq \kappa V_3^{p_2/p_1}(j)$
			for some $\kappa\le 1/2$.
			
			Finally, by Theorem 2.1 from \cite{tuominen1994subgeometric} we get
			$$\lim_{t\to\infty} e^{\rho t}\|P(t, j, \cdot)-\bmu^*(\cdot)\|=0,\, \forall\, j:=(s,i,q,r,k)\in\R^{4,*}_+\times\M,\,\text{ for some constant }\,\rho>0$$
			and the proof is complete.
			
	\end{proof}

	\section{Simulations and Discussion}\label{sec:disc}
	\subsection{Comparison with the literature.}
	There are a plethora of results for stochastic epidemic models with various types of noise. However, most of the available results have artificial assumptions and do not provide sharp results. In our paper we can consider a large class of models and are able to prove under natural assumptions that a single parameter, $\lambda$, determines the long term fate of the disease. A paper that is closely related to ours is \cite{Zhou21} where the authors study a SIQRS model when there is only white noise. The authors showcase conditions under which the pandemic persists but do not provide any results for extinction. Moreover, their results are not sharp because they do not establish necessary and sufficient conditions for persistence. Nevertheless, \cite{Zhou21} provided the inspiration for the current paper. 
	\subsection{Linear incidence rate.}
	In this section we will assume that $F(S,I)=\beta S$.
	\subsubsection{The deterministic model.}
	The behavior of \eqref{ww0} is well known. Set 
	\[
	R_0 = \frac{A\beta}{\mu(\mu+\gamma_1+\gamma_2)}.
	\]
	Using the results from \cite{Zhou21} one can show that if $R_0>1$ then there is a globally stable equilibrium $E$ in $(0,\infty)^4$ and the disease persists. Instead, if $R_0\leq 1$ the disease goes extinct. 
	This expression shows that if $\gamma_2$ is large it is harder for the disease to persist. If $\gamma_1$ is large it is also harder for the disease to persist but this would mean the disease kills people very fast and therefore cannot spread. 
	\subsubsection{The effects of white noise.}
	Note that if there is no switching and $F(S,I)=\beta S$ our criterion boils down to checking whether
	\begin{equation}\label{e_lam1}
		\lambda = \frac{A\beta}{\mu} - (\mu+\gamma_1+\gamma_2) - \frac{\sigma_2^2}{2}
	\end{equation}
	is positive or negative. In this setting it seems that the environmental fluctuations are detrimental to the spread of the disease since when $\lambda<0$ the disease dies out. Only when $\lambda>0$ can the disease persist long term. We note that in \cite{Zhou21} the authors prove that the disease persists when
	\[
	R:=\frac{A\beta}{\mu+\frac{\sigma_1^2}{2}} - (\mu+\gamma_1+\gamma_2) - \frac{\sigma_2^2}{2}>0.
	\]
	However, since $R<\lambda$, this result is not sharp. If  $R < 0 < \lambda$ the results from \cite{Zhou21} are not able to show the persistence of the disease, whereas our results do show that the disease persists.
	
	\subsubsection{The impact of quarantine.}
	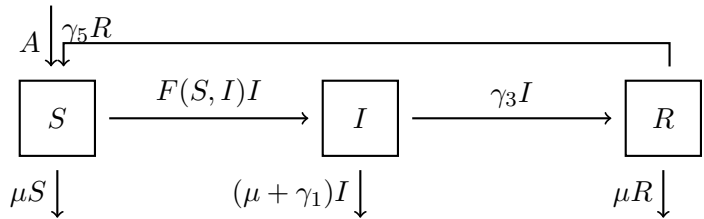
\begin{figure}
		\begin{center}
			\begin{tikzpicture}[
				skip loop/.style={to path={-- ++(0,#1) -| (\tikztotarget)}},
				thick,
				nonterminal/.style = {
					rectangle,			
					minimum size = 10mm,			
					thick,
					draw = black,
				}
				]
				\matrix [row sep=10mm, column sep=15mm] {
					\node (pa){}; &&&&&&\\
					\node (s) [nonterminal] {$S$}; &&
					\node (i) [nonterminal] {$I$}; && 
					\node (r) [nonterminal] {$R$}; \\
					\node (pm1){}; && 
					\node (pmg1){}; && 
					\node (pm2){}; \\ 
				};
				\path (pa) edge[->, shorten >=5pt, transform canvas={xshift=-0.5ex}] node[left]{$A$} node[right, transform canvas={yshift=1ex}]{$\gamma_5 R$} (s);
				\path (s) edge[->, shorten <=5pt, shorten >=5pt] node[above]{$F(S,I) I$} (i);
				\path (i) edge[->, shorten <=5pt, shorten >=5pt] node[above]{$\gamma_3 I$} (r);
				\path (r) edge[->, shorten <=5pt, shorten >=5pt, transform canvas={xshift=0.5ex}, skip loop=10mm ] node[right]{} (s);
				\path (s) edge[->, shorten <=5pt, shorten >=5pt] node[left]{$\mu S$}(pm1);
				\path (i) edge[->, shorten <=5pt, shorten >=5pt] node[left]{$(\mu + \gamma_1) I$}(pmg1);
				\path (r) edge[->, shorten <=5pt, shorten >=5pt] node[left]{$\mu R$}(pm2);
			\end{tikzpicture}
		\end{center}
                \caption{Schematic representation of the counterfactual SIRS system used in numerical experiments in Section \ref{sec:disc}.}
		\label{fig:syst_2}
	\end{figure}
	
        \begin{figure}
            \includegraphics[width=0.45\textwidth]{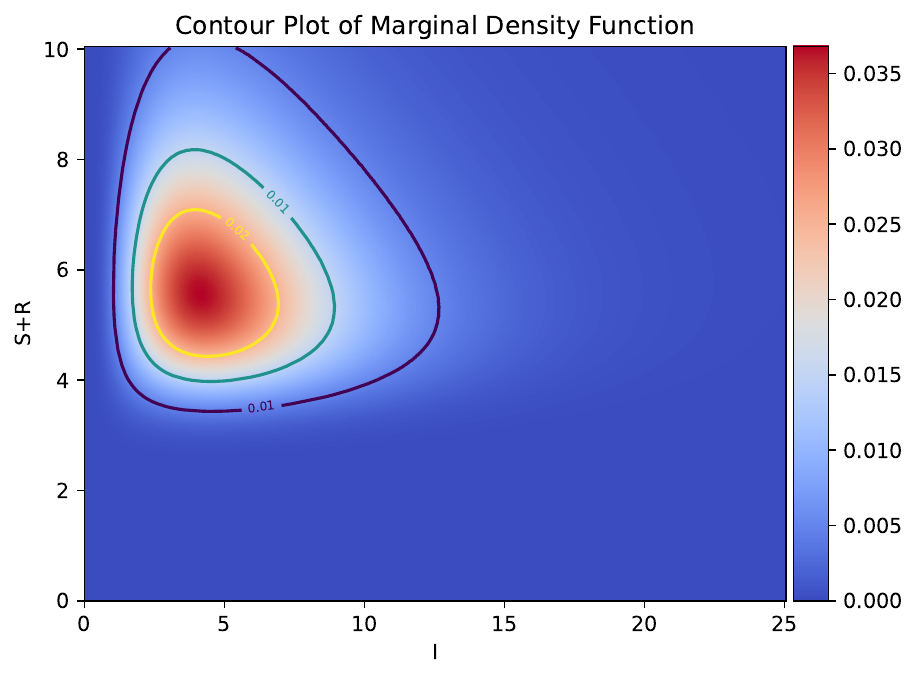}
            \includegraphics[width=0.45\textwidth]{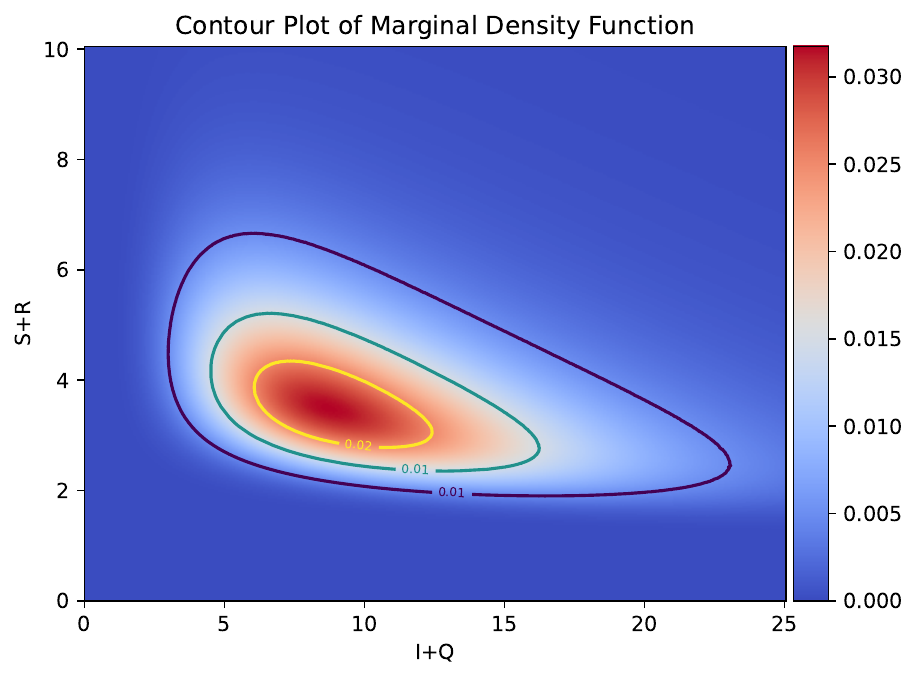}
            \caption{Side by side comparison of the estimation of two SSDE models, one without a quarantined group (SIRS) on the left, and one with (SIQRS), on the right.
            Both models have a linear infection incidence rate, given by $F(S, I) = \beta S$.
            The parameters for the SIRS  model are $A = 20$, $\beta = 1$, $\sigma_. = 1$, $\mu = \gamma_1 = \gamma_3 = \gamma_5 = 1$.
            For the SIQRS model, the other parameters are the same, but now $\gamma_1 = 0$ and $\gamma_2 = \gamma_4 = 1$.
            Using \eqref{e_lam1} we note that for both models $\lambda = 16.5 > 0$.
            The approximate population expectations at stationarity for the SIRS model are $\E\overline{S} \approx 3.5$, $\E\overline{I} \approx 6.6$, and $\E\overline{R} \approx 3.3$.
            The expectation of the total population is $\approx 13.4$.
            In the SIQRS model, the approximate population expectations at stationarity are $\E\overline{S} \approx 2.5$, $\E\overline{I} \approx 9.5$, $\E\overline{Q} \approx 3.2$, and $\E\overline{R} \approx 1.6$.
            The expectation of the total population is $\approx 16.8$.
            }
            \label{fig:contour_1}
        \end{figure}
	
	We simulate two simple models with linear incidence and no switching, one with a quarantined group and one without, to determine the impact of introducing a quarantined group.
	The crucial assumption is that the excess mortality is shifted to the quarantined group in the quarantine model.
	No other change is introduced.
	In particular, the quarantined group \emph{does not} have an improved mortality rate.
	
	The diagram in Figure \ref{fig:syst} summarises the model and the parameters used in the SIQRS model.
        Figure \ref{fig:syst_2} summarises the counterfactual SIRS model that we will use to contrast the effect of adding a quarantined group.
	Figure \ref{fig:contour_1} shows the marginal probability density functions at stationarity. We focus on population composites of interest: $S+R$, the healthy group, and $I$ or $I + Q$, the sick group.
	An approximation for the stationary distribution is obtained using the Monte Carlo method with the Milstein refinement \cite{giles2006improved}. Note that by slightly modifying Theorem \ref{t:pers} we can show that the randomized occupation measures (empirical measures) converge exponentially fast to the stationary distribution. This implies that MC methods will work well.
	
	The unusual conclusion that we can draw is that simply isolating individuals in a quarantined group -- without providing any improvements in treatment leading to decreased mortality -- can increase the population overall, at the expense of having a larger share of infected individuals.
	This can intuitively be understood by noting that the rate of movement into the $I$ category is in proportion to the population size already in the $I$ category.
	By making individuals non-infective -- i.e., by moving them into a category $Q$ -- we are therefore reducing the burden of the disease on total population.

	\subsubsection{The effects of switching.}
        Figure \ref{fig:switching} shows a plot where a non-trivial effect of switching is added to the simple linear SIQRS model from Figure \ref{fig:contour_1}.
        Now the parameter $\beta$ switches between the values $\{1, 4\}$, according to the switching rate $q_{ij} = q = 0.1$ ($i \neq j$).
        We can observe that the marginal density cross-section shows a bimodal density.
        Numerical experiments suggest that this is a common occurrence in models with switching, as the density tends to approximate an overlap of the densities of two simpler models without switching -- in this case two models with fixed parameters $\beta = 1$ and $\beta = 4$.
        \begin{figure}
            \includegraphics[width=0.45\textwidth]{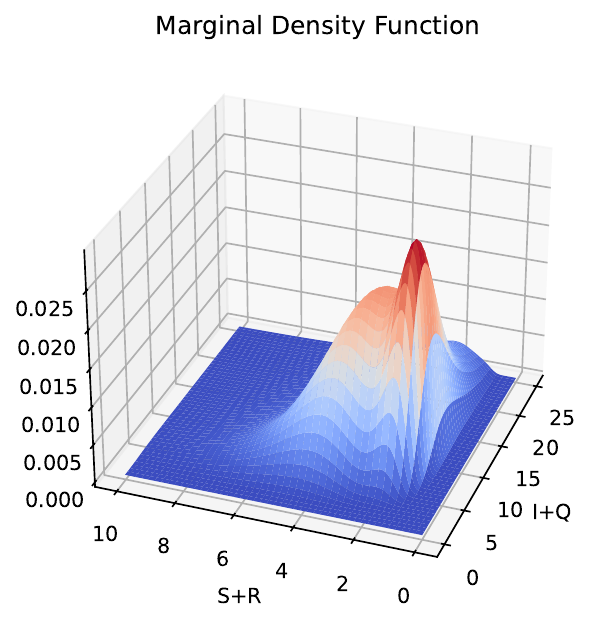}
            \includegraphics[width=0.45\textwidth]{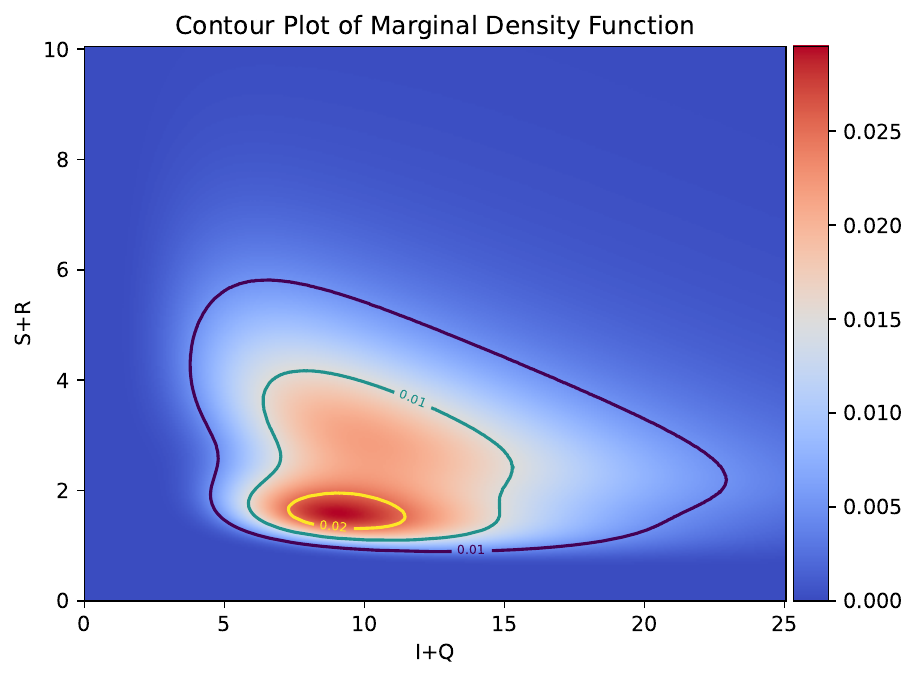}
            \caption{The estimation of a SIQRS SSDE model with switching, on the left the 3D plot of marginal density of a cross-section, and on the right the contour equivalent plot.
            Note the fact that the density is bimodal.
            The model has a linear infection incidence rate, given by $F(S, I) = \beta S$.
            The parameter values are those summarised in the caption of Figure \ref{fig:contour_1}, but with the exception that now $\beta$ switches between the values $1$ and $4$ according to the switching rate $q_{ij} = q = 0.1$.
            The approximate population expectations at stationarity for the model are $\E \overline{S} \approx 1.6$, $\E\overline{I} \approx 10.1$, and $\E\overline{Q} \approx 3.4$, and $\E\overline{R} \approx 1.7$.
            The expectation of the total population is $\approx 16.6$.
            }
            \label{fig:switching}
        \end{figure}
	
	\subsection{Nonlinear incidence rates.} We note that most of papers which look at related SIQS/SIQRS models are only able to treat linear incidence rates \cite{WL23, Zhou21}, and thus do not capture some of the most important models. The seminal paper \cite{LLI86} has shown that non-linearities appear realistically and that the dynamics changes significantly, even in the absence of noise. In \cite{ruan2003dynamical} the authors show that stable and unstable limit cycles can appear in SIRS models with nonlinear incidence rates. 
	
	\subsubsection{Holling type II functional response.}
	We explore in our simulations other types of incidence rates used in the wider literature, for example $F(S, I) = \beta S/(1 + m_1S)$, also known as Holling type II.
	Figure \ref{fig:contour_2} shows that the effect of introducing a quarantined group in such a model with no switching is broadly similar to the linear response model.
        That is, we see an increase in the total average population, but at the expense of having a larger fraction of sick individuals.
	
        \begin{figure}
            \includegraphics[width=0.45\textwidth]{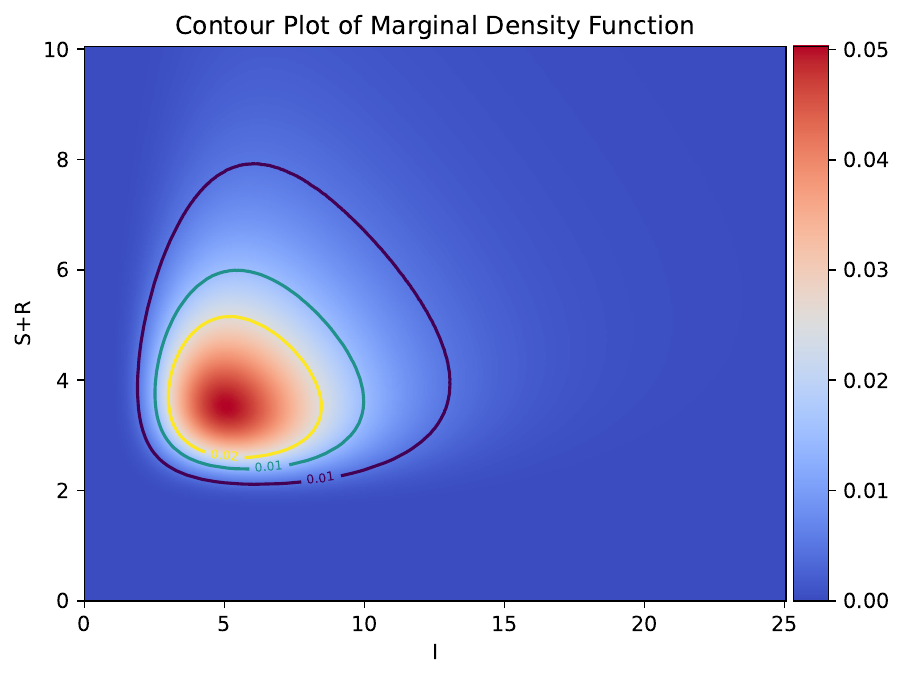}
            \includegraphics[width=0.45\textwidth]{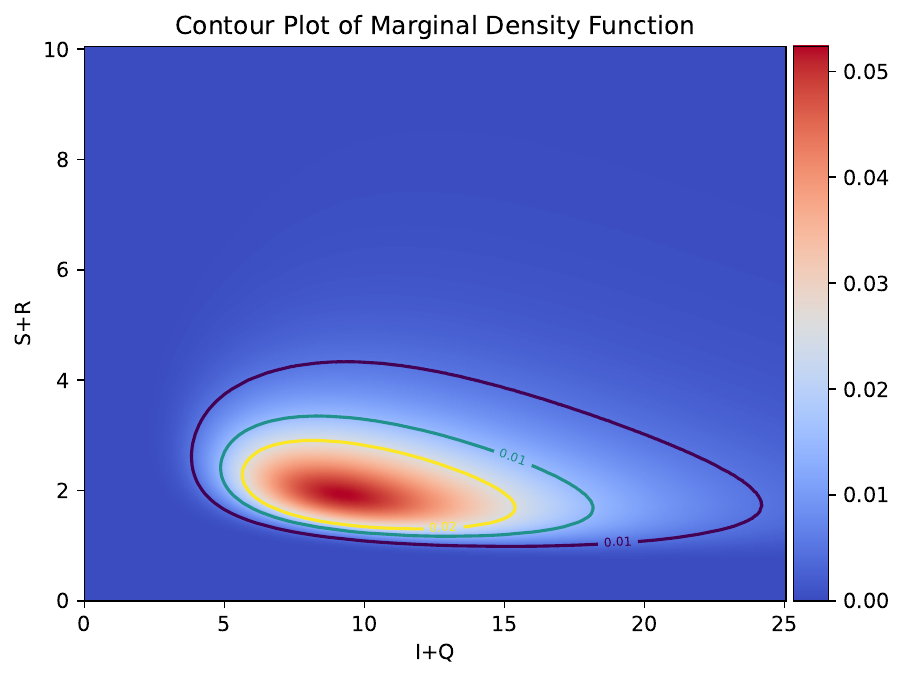}
            \caption{Side by side comparison of the estimation of two SSDE models, one without a quarantined group (SIRS) on the left, and one with (SIQRS), on the right.
            Both models have a Holling type II infection incidence rate, given by $F(S, I) = \beta S/(1 + m_1 S)$, where $m_1 = 0.1$.
            The other parameters for the model, summarised in Figure \ref{fig:syst}, are $A = 20$, $\beta = 3$, $\sigma_. = 1$, $\mu = \gamma_1 = \gamma_3 = \gamma_5 = 1$.
            For the SIQRS model, the other parameters are the same, but now $\gamma_1 = 0$ and $\gamma_2 = \gamma_4 = 1$.
            The approximate population expectations at stationarity for the SIRS model are $\E \overline{S} \approx 1.4$, $\E\overline{I} \approx 7.5$, and $\E\overline{R} \approx 3.7$.
            The expectation of the total population is $\approx 12.5$.
            In the SIQRS model, the approximate population expectations at stationarity are $\E\overline{S} \approx 0.9$, $\E\overline{I} \approx 10.4$, $\E\overline{Q} \approx 3.5$, and $\E\overline{R} \approx 1.7$.
            The expectation of the total population is $\approx 16.5$.
            }
            \label{fig:contour_2}
        \end{figure}

	\subsubsection{Beddington-DeAngelis functional response.}
        \begin{figure}
            \includegraphics[width=0.45\textwidth]{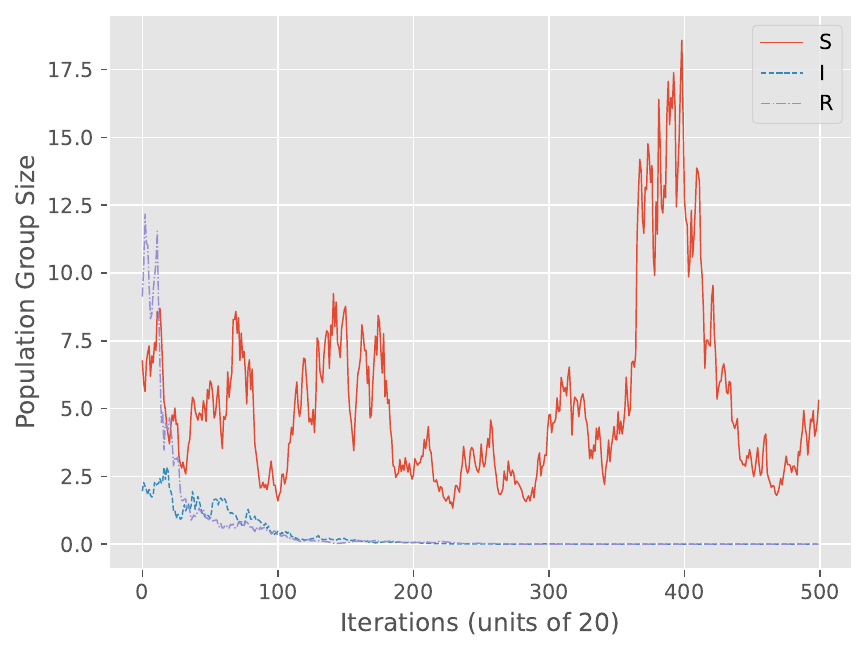}
            \includegraphics[width=0.45\textwidth]{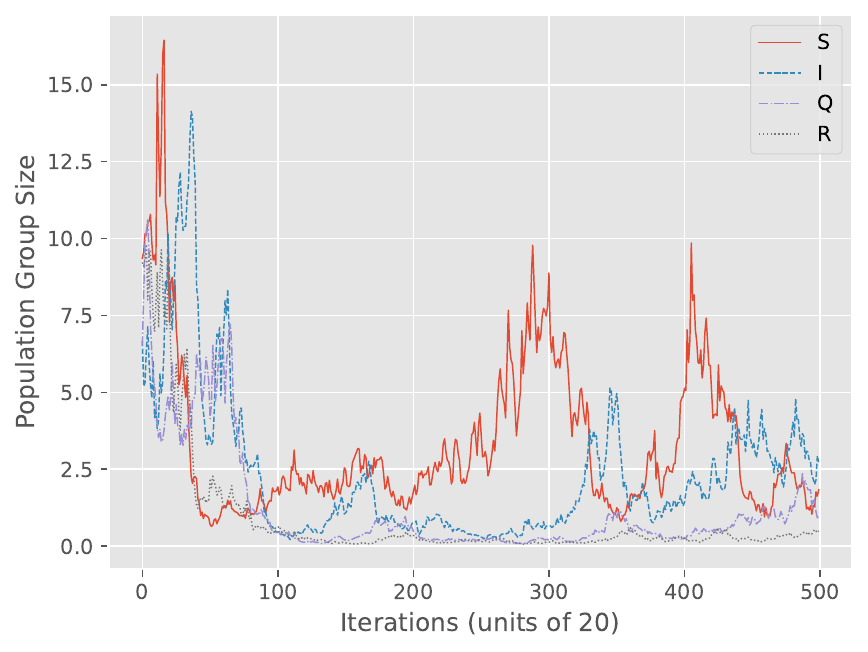}
            \caption{Side by side comparison of the estimation of two SSDE models, one without a quarantined group (SIRS) on the left, and one with (SIQRS), on the right.
            Both models have a Beddington-DeAngelis infection incidence rate, given by $F(S, I) = \beta S/(1 + m_1 S + m_2 I)$, where $m_1 = m_2 = 0.1$.
            The common parameters for the two models are $A = 5$, $\beta = 1$, $\sigma_. = 1$, $\mu = \gamma_1 = \gamma_3 = \gamma_5 = 1$.
            For the SIQRS model, the other parameters are the same, but now $\gamma_1 = 0$ and $\gamma_2 = \gamma_4 = 1$.
            The horizontal represents periods of simulation, with one plotted point for every $20$ periods.
            One period of simulation corresponds to a discrete increment of time $\Delta t = 0.001$.
            Note that $\lambda = -0.1\overline{6}$ for the SIRS model on the left, and $\lambda = 0.8\overline{3}$ for the SIQRS model on the right.
            The long term average of the total population for the SIRS model is $\approx 5.0$.
            For the SIQRS model, it is $\approx 4.8$.
            }
            \label{fig:timeseries_1}
        \end{figure}
	$F(S, I) = \beta S/(1 + m_1 S + m_2 I)$ 
	For both the SIQRS and SIRS models we find using \cite{D18} that, in the absence of switching,
	\[
	\lambda = \frac{A\beta}{\mu + m_1 A} -c_2-\frac{\sigma_2^2}{2}.
	\]
        We estimate the time evolution of such models in Figure \ref{fig:timeseries_1}.
        The parameters are such that we have a positive $\lambda$ for the SIRS model, which changes to a slight negative $\lambda$ for the SIQRS model.
        The number of infected and recovered goes to $0$ in the SIRS model, but stays positive in the SIQRS model, since $c_2$ is decreased by the value change of $\gamma_1$, from $1$ to $0$.
        Observe that, in this case, the introduction of a quarantined group has a negative effect on the long term prevalence of the disease.
        The long term average of the population is also slightly reduced in the model with quarantine. This shows how quarantine can be beneficial for certain diseases.

\subsection{Alternative formulations.}
In the previous experiments, we assumed a particular way in which the introduction of the quarantined group `affects' the baseline SIRS system.
Specifically, we have shifted $\gamma_1$ to $0$, to capture the somewhat strong assumption that excess mortality appears only in the quarantined group.
In other words, all members of the population are quarantined before they suffer greatly from the disease.
There are no other changes in parameters in our contrasting comparisons, although we could have made other parameter changes, say justified on empirical grounds.

Another debatable modeling choice is in the way the quarantined group is a staging group for all infected individuals before they reach the recovered group.
This excludes applications in which the quarantine is implemented imperfectly, or in which the disease leads to infected individuals that are able to spread the disease but fall below a certain diagnosis threshold, before fully recovering.
E.g., the symptoms of a diseased individual may be too light to lead to detection and quarantine.
To include such set-ups, we would need a branching in the way the individuals move between the groups.
 \clearpage 
\bibliographystyle{amsalpha}
\bibliography{SIQRS_paper_bibliography}

\newcommand{\etalchar}[1]{$^{#1}$}
\providecommand{\bysame}{\leavevmode\hbox to3em{\hrulefill}\thinspace}
\providecommand{\MR}{\relax\ifhmode\unskip\space\fi MR }
\providecommand{\MRhref}[2]{%
  \href{http://www.ams.org/mathscinet-getitem?mr=#1}{#2}
}
\providecommand{\href}[2]{#2}
\begin{thebibliography}{MNR{\etalchar{+}}21}

\bibitem[AB12]{AB12}
Hakan Andersson and Tom Britton, \emph{Stochastic epidemic models and their
  statistical analysis}, vol. 151, Springer Science \& Business Media, 2012.

\bibitem[Ben18]{B18}
M.~Bena{\"\i}m, \emph{Stochastic persistence}, preprint.

\bibitem[BHS08]{BHS08}
M.~Bena{\"\i}m, J.~Hofbauer, and W.~H. Sandholm, \emph{Robust permanence and
  impermanence for stochastic replicator dynamics}, J. Biol. Dyn. \textbf{2}
  (2008), no.~2, 180--195. \MR{2427526}

\bibitem[BL16]{BL16}
M.~Bena{\"\i}m and C.~Lobry, \emph{Lotka {V}olterra in fluctuating environment
  or ``how switching between beneficial environments can make survival
  harder''}, The Annals of Applied Probability \textbf{26} (2016), no.~6,
  3754--3785.

\bibitem[BS09]{BS09}
M.~Bena{\"\i}m and S.~J. Schreiber, \emph{Persistence of structured populations
  in random environments}, Theoretical Population Biology \textbf{76} (2009),
  no.~1, 19--34.

\bibitem[Che82]{C82}
Peter~L Chesson, \emph{The stabilizing effect of a random environment}, Journal
  of Mathematical Biology \textbf{15} (1982), no.~1, 1--36.

\bibitem[Dav84]{D84}
M.~H.~A. Davis, \emph{Piecewise-deterministic markov processes: A general class
  of non-diffusion stochastic models}, Journal of the Royal Statistical
  Society: Series B (Methodological) \textbf{46} (1984), no.~3, 353--376.

\bibitem[DDN19]{du2019conditions}
Nguyen~Huu Du, Nguyen~Thanh Dieu, and Nguyen~Ngoc Nhu, \emph{Conditions for
  permanence and ergodicity of certain sir epidemic models}, Acta Applicandae
  Mathematicae \textbf{160} (2019), 81--99.

\bibitem[Die18]{D18}
Nguyen~Thanh Dieu, \emph{Asymptotic properties of a stochastic sir epidemic
  model with beddington--deangelis incidence rate}, Journal of Dynamics and
  Differential Equations \textbf{30} (2018), no.~1, 93--106.

\bibitem[DN17]{du2017permanence}
Nguyen~Huu Du and Nguyen~Ngoc Nhu, \emph{Permanence and extinction of certain
  stochastic sir models perturbed by a complex type of noises}, Applied
  Mathematics Letters \textbf{64} (2017), 223--230.

\bibitem[DNDY16]{dieu2016classification}
Nguyen~Thanh Dieu, Dang~Hai Nguyen, Nguyen~Huu Du, and G~Yin,
  \emph{Classification of asymptotic behavior in a stochastic sir model}, SIAM
  Journal on Applied Dynamical Systems \textbf{15} (2016), no.~2, 1062--1084.

\bibitem[EHS15]{EHS15}
S.~N. Evans, A.~Hening, and S.~J. Schreiber, \emph{Protected polymorphisms and
  evolutionary stability of patch-selection strategies in stochastic
  environments}, J. Math. Biol. \textbf{71} (2015), no.~2, 325--359.
  \MR{3367678}

\bibitem[ERSS13]{ERSS13}
S.~N. Evans, P.~L. Ralph, S.~J. Schreiber, and A.~Sen, \emph{Stochastic
  population growth in spatially heterogeneous environments}, J. Math. Biol.
  \textbf{66} (2013), no.~3, 423--476. \MR{3010201}

\bibitem[FT95]{FT95}
Zhilan Feng and Horst~R Thieme, \emph{Recurrent outbreaks of childhood diseases
  revisited: the impact of isolation}, Mathematical biosciences \textbf{128}
  (1995), no.~1-2, 93--130.

\bibitem[FT00]{FT00}
\bysame, \emph{Endemic models with arbitrarily distributed periods of infection
  ii: fast disease dynamics and permanent recovery}, SIAM Journal on Applied
  Mathematics \textbf{61} (2000), no.~3, 983--1012.

\bibitem[Gil06]{giles2006improved}
Mike Giles, \emph{Improved multilevel monte carlo convergence using the
  milstein scheme}, Monte Carlo and Quasi-Monte Carlo Methods 2006, Springer,
  2006, pp.~343--358.

\bibitem[HHNN23]{HN23}
Alexandru Hening, Nguyen~Trong Hieu, Dang~Hai Nguyen, and Nhu~Ngoc Nguyen,
  \emph{Stochastic nutrient-plankton models}, Journal of Differential Equations
  \textbf{376} (2023), 370--405.

\bibitem[HMT09]{huang2009global}
Gang Huang, Wanbiao Ma, and Yasuhiro Takeuchi, \emph{Global properties for
  virus dynamics model with beddington--deangelis functional response}, Applied
  Mathematics Letters \textbf{22} (2009), no.~11, 1690--1693.

\bibitem[HN18]{hening2018coexistence}
A.~Hening and D.~Nguyen, \emph{Coexistence and extinction for stochastic
  {K}olmogorov systems}, The Annals of applied probability \textbf{28} (2018),
  no.~3, 1893--1942.

\bibitem[HN20]{HN20}
A.~Hening and D.~H. Nguyen, \emph{The competitive exclusion principle in
  stochastic environments}, Journal of Mathematical Biology \textbf{80} (2020),
  1323–--1351.

\bibitem[HNC21]{HNC20}
A~Hening, D.~Nguyen, and P~Chesson, \emph{A general theory of coexistence and
  extinction for stochastic ecological communities}, Journal of Mathematical
  Biology \textbf{82} (2021), no.~6, 1--76.

\bibitem[HNS22]{HNS21}
A.~Hening, D.~H. Nguyen, and S.~J. Schreiber, \emph{A classification of the
  dynamics of three-dimensional stochastic ecological systems}, Annals of
  Applied Probability \textbf{32} (2022), no.~2.

\bibitem[HZS02]{H02}
Herbert Hethcote, Ma~Zhien, and Liao Shengbing, \emph{Effects of quarantine in
  six endemic models for infectious diseases}, Mathematical biosciences
  \textbf{180} (2002), no.~1-2, 141--160.

\bibitem[JGH{\etalchar{+}}16]{jiang2016dynamical}
Jiehui Jiang, Sixing Gong, Bing He, et~al., \emph{Dynamical behavior of a rumor
  transmission model with holling-type ii functional response in emergency
  event}, Physica A: Statistical Mechanics and its Applications \textbf{450}
  (2016), 228--240.

\bibitem[KM27]{kermack1927contribution}
William~Ogilvy Kermack and Anderson~G McKendrick, \emph{A contribution to the
  mathematical theory of epidemics}, Proceedings of the royal society of
  london. Series A, Containing papers of a mathematical and physical character
  \textbf{115} (1927), no.~772, 700--721.

\bibitem[KM32]{kermack1932contributions}
\bysame, \emph{Contributions to the mathematical theory of epidemics. ii.—the
  problem of endemicity}, Proceedings of the Royal Society of London. Series A,
  containing papers of a mathematical and physical character \textbf{138}
  (1932), no.~834, 55--83.

\bibitem[LES03]{LES03}
R.~Lande, S.~Engen, and B.-E. Saether, \emph{Stochastic population dynamics in
  ecology and conservation}, Oxford University Press on Demand, 2003.

\bibitem[LJ14]{lin2014threshold}
Yuguo Lin and Daqing Jiang, \emph{Threshold behavior in a stochastic sis
  epidemic model with standard incidence}, Journal of Dynamics and Differential
  Equations \textbf{26} (2014), 1079--1094.

\bibitem[LJHA18]{liu2018analysis}
Qun Liu, Daqing Jiang, Tasawar Hayat, and Bashir Ahmad, \emph{Analysis of a
  delayed vaccinated sir epidemic model with temporary immunity and l{\'e}vy
  jumps}, Nonlinear Analysis: Hybrid Systems \textbf{27} (2018), 29--43.

\bibitem[LLI86]{LLI86}
Wei-min Liu, Simon~A Levin, and Yoh Iwasa, \emph{Influence of nonlinear
  incidence rates upon the behavior of sirs epidemiological models}, Journal of
  mathematical biology \textbf{23} (1986), 187--204.

\bibitem[LLWZ19]{lan2019stochastic}
Guijie Lan, Ziyan Lin, Chunjin Wei, and Shuwen Zhang, \emph{A stochastic sirs
  epidemic model with non-monotone incidence rate under regime-switching},
  Journal of the Franklin Institute \textbf{356} (2019), no.~16, 9844--9866.

\bibitem[MNR{\etalchar{+}}21]{M21}
Shiva Moein, Niloofar Nickaeen, Amir Roointan, Niloofar Borhani, Zarifeh
  Heidary, Shaghayegh~Haghjooy Javanmard, Jafar Ghaisari, and Yousof Gheisari,
  \emph{Inefficiency of sir models in forecasting covid-19 epidemic: a case
  study of isfahan}, Scientific reports \textbf{11} (2021), no.~1, 4725.

\bibitem[NNY20a]{nguyen2020analysis}
Dang~H Nguyen, Nhu~N Nguyen, and George Yin, \emph{Analysis of a spatially
  inhomogeneous stochastic partial differential equation epidemic model},
  Journal of Applied Probability \textbf{57} (2020), no.~2, 613--636.

\bibitem[NNY20b]{nguyen2020general}
\bysame, \emph{General nonlinear stochastic systems motivated by chemostat
  models: Complete characterization of long-time behavior, optimal controls,
  and applications to wastewater treatment}, Stochastic Processes and their
  Applications \textbf{130} (2020), no.~8, 4608--4642.

\bibitem[NSY18]{nie2018dynamic}
Lin-Fei Nie, Jing-Yun Shen, and Chen-Xia Yang, \emph{Dynamic behavior analysis
  of sivs epidemic models with state-dependent pulse vaccination}, Nonlinear
  Analysis: Hybrid Systems \textbf{27} (2018), 258--270.

\bibitem[NY19]{nguyen2019stochastic}
Nhu~N Nguyen and George Yin, \emph{Stochastic partial differential equation sis
  epidemic models: modeling and analysis}, Communications on Stochastic
  Analysis \textbf{13} (2019), no.~3, 8.

\bibitem[NYZ20]{nguyen2020long}
Dang~Hai Nguyen, George Yin, and Chao Zhu, \emph{Long-term analysis of a
  stochastic sirs model with general incidence rates}, SIAM Journal on Applied
  Mathematics \textbf{80} (2020), no.~2, 814--838.

\bibitem[POT20]{phu2020longtime}
Nguyen~Dinh Phu, Donal O’Regan, and Tran~Dinh Tuong, \emph{Longtime
  characterization for the general stochastic epidemic sis model under
  regime-switching}, Nonlinear Analysis: Hybrid Systems \textbf{38} (2020),
  100951.

\bibitem[RW03]{ruan2003dynamical}
Shigui Ruan and Wendi Wang, \emph{Dynamical behavior of an epidemic model with
  a nonlinear incidence rate}, Journal of differential equations \textbf{188}
  (2003), no.~1, 135--163.

\bibitem[SBA11]{SBA11}
S.~J. Schreiber, M.~Bena{\"\i}m, and K.~A.~S. Atchad{\'e}, \emph{Persistence in
  fluctuating environments}, J. Math. Biol. \textbf{62} (2011), no.~5,
  655--683. \MR{2786721}

\bibitem[SLS09]{SLS09}
S.~J. Schreiber and J.~O. Lloyd-Smith, \emph{Invasion dynamics in spatially
  heterogeneous environments}, The American Naturalist \textbf{174} (2009),
  no.~4, 490--505.

\bibitem[THRZ08]{tang2008coexistence}
Yilei Tang, Deqing Huang, Shigui Ruan, and Weinian Zhang, \emph{Coexistence of
  limit cycles and homoclinic loops in a sirs model with a nonlinear incidence
  rate}, SIAM Journal on Applied Mathematics \textbf{69} (2008), no.~2,
  621--639.

\bibitem[TT94]{tuominen1994subgeometric}
Pekka Tuominen and Richard~L Tweedie, \emph{Subgeometric rates of convergence
  of f-ergodic markov chains}, Advances in Applied Probability \textbf{26}
  (1994), no.~3, 775--798.

\bibitem[WL23]{WL23}
Feng Wang and Zaiming Liu, \emph{Dynamical behavior of a stochastic siqs model
  via isolation with regime-switching}, Journal of Applied Mathematics and
  Computing \textbf{69} (2023), no.~2, 2217--2237.

\bibitem[ZHXM13]{zhang2013sirs}
Xiao-Bing Zhang, Hai-Feng Huo, Hong Xiang, and Xin-You Meng, \emph{An sirs
  epidemic model with pulse vaccination and non-monotonic incidence rate},
  Nonlinear Analysis: Hybrid Systems \textbf{8} (2013), 13--21.

\bibitem[ZJDH21]{Zhou21}
Baoquan Zhou, Daqing Jiang, Yucong Dai, and Tasawar Hayat, \emph{Stationary
  distribution and density function expression for a stochastic siqrs epidemic
  model with temporary immunity}, Nonlinear Dynamics \textbf{105} (2021),
  no.~1, 931--955.

\bibitem[ZWZ15]{zhang2015threshold}
Liang Zhang, Zhi-Cheng Wang, and Xiao-Qiang Zhao, \emph{Threshold dynamics of a
  time periodic reaction--diffusion epidemic model with latent period}, Journal
  of Differential Equations \textbf{258} (2015), no.~9, 3011--3036.

\end{thebibliography}

\end{document}